\documentclass[lettersize,journal]{IEEEtran}

\usepackage{cite}
\usepackage{amsmath,amssymb,amsfonts}
\usepackage{algorithmic}
\usepackage{graphicx}
\usepackage{textcomp}
\usepackage{subfig}
\usepackage{dsfont}
\usepackage{amsthm}
\usepackage{color, soul}
\usepackage{hyperref}

\def\BibTeX{{\rm B\kern-.05em{\sc i\kern-.025em b}\kern-.08em
    T\kern-.1667em\lower.7ex\hbox{E}\kern-.125emX}}
\renewcommand{\arraystretch}{2.2}
\usepackage{array}
\newcolumntype{M}[1]{>{\centering\arraybackslash}m{#1}}
\usepackage{hhline}

\newcommand{\pt}{ \widetilde{p}}
\newcommand{\Yt}{ \widetilde{Y}}
\newcommand{\yt}{ \widetilde{y}}
\newcommand{\Yhn}{\hat{Y}_{Nr2F}}
\newcommand{\Yhr}{\hat{Y}_{RSSDU}}
\newcommand{\Mt}{M_{\Lambda}}
\newcommand{\Mo}{M_{\Omega}}

\newcommand{\mt}{ \widetilde{m}}
\newcommand{\Pds}{ \mathds{P}}
\newcommand{\Eds}{ \mathds{E}}
\newcommand{\eye}{ \mathds{1}}
\newcommand{\Nt}{ \widetilde{N}}
\newcommand{\nt}{ \widetilde{n}}

\newcommand{\Mlo}{M_{\Lambda \cap \Omega} }
\newcommand{\Moml}{M_{\Omega \setminus \Lambda } }

\newcommand{\Mloc}{M_{ (\Lambda \cap \Omega  )^c }} 
 
\newcommand{\Moc}{M_{ \Omega ^c }}

\newcommand{\Pc}{\mathcal{P}}

\newcommand{\dt}{\nabla_\theta}

\newcommand{\A}{\frac{1+\alpha^2}{\alpha^2}}

\newcommand{\ft}{f_\theta ( \Yt)}

\newtheorem{claim}{Claim}
\newtheorem{lemma}{Lemma}
\newtheorem{result}{Result}

\begin{document}
\title{Clean self-supervised MRI reconstruction from noisy, sub-sampled training data with Robust SSDU}
\author{Charles Millard and Mark Chiew 
\thanks{This work was supported in part by the Engineering and Physical Sciences Research Council, grant EP/T013133/1, by the Royal Academy of Engineering, grant RF201617/16/23, and by the Wellcome Trust, grant 203139/Z/16/Z. The computational aspects of this research were supported by the Wellcome Trust Core Award Grant Number 203141/Z/16/Z and the NIHR Oxford BRC. The views expressed are those of the authors and not necessarily those of the NHS, the NIHR or the Department of Health. This research was undertaken, in part, thanks to funding from the Canada Research Chairs Program.}
\thanks{Charles Millard and Mark Chiew are with the Wellcome Centre for Integrative Neuroimaging, FMRIB, Nuffield Department of Clinical Neurosciences, University of Oxford, Level 0, John Radcliffe Hospital, Oxford, OX3 9DU, UK. Mark Chiew is also with the Department of Medical Biophysics, University of Toronto, Toronto, Canada and Physical Sciences, Sunnybrook Research Institute, Toronto, Canada (email: charles.millard@ndcn.ox.ac.uk and mark.chiew@utoronto.ca).}}

\markboth{IEEE TRANSACTIONS ON COMPUTATIONAL IMAGING, VOL. X, NO. X, XXX XX}{Self-supervised deep learning MRI reconstruction with Noisier2Noise}

\maketitle

\begin{abstract}
Most existing methods for Magnetic Resonance Imaging (MRI) reconstruction with deep learning use fully supervised training, which assumes that a high signal-to-noise ratio (SNR), fully sampled dataset is available for training. In many circumstances, however, such a dataset is highly impractical or even technically infeasible to acquire. Recently, a number of self-supervised methods for MR reconstruction have been proposed, which use sub-sampled data only. However, the majority of such methods, such as Self-Supervised Learning via Data Undersampling (SSDU), are susceptible to reconstruction errors arising from noise in the measured data. In response, we propose Robust SSDU, which provably recovers clean images from noisy, sub-sampled training data by simultaneously estimating missing k-space samples and denoising the available samples. Robust SSDU trains the reconstruction network to map from a further noisy and sub-sampled  version of the data to the original, singly noisy and sub-sampled data, and applies an additive Noisier2Noise correction term at inference. We also present a related method, Noiser2Full, that recovers clean images when noisy, fully sampled data is available for training. Both proposed methods are applicable to any network architecture, straight-forward to implement and have similar computational cost to standard training.  We evaluate our methods on the multi-coil fastMRI brain dataset with a novel denoising-specific architecture and find that it performs competitively with a benchmark trained on clean, fully sampled data.
\end{abstract}




\section{Introduction}

Magnetic Resonance Imaging (MRI) has excellent soft tissue contrast and is the gold standard  modality for a number of clinical applications. A hindrance of MRI, however, is its lengthy acquisition time, which is especially challenging when high spatio-temporal resolution is required, such as for dynamic imaging \cite{bustin2020compressed}. To address this, there has been substantial research attention on methods that reduce the acquisition time without significantly sacrificing the diagnostic  quality \cite{Pruessmann1999, Lustig2007, Ye2019}. In MRI, measurements are acquired in the Fourier representation of the image, referred to in the MRI literature as ``k-space". Since the acquisition time is roughly proportional to the number of k-space samples, acquisitions can be accelerated by sub-sampling.  A reconstruction algorithm is then employed to estimate the image from the sub-sampled data.


In recent years, reconstructing sub-sampled MRI data with neural networks has emerged as the state-of-the-art \cite{wang2016accelerating, kwon2017parallel, hammernik2018learning}. The majority of existing methods assume that a fully sampled dataset is available for fully-supervised training. However, for many applications,  no such dataset is available, and may be difficult or even infeasible to acquire in practice \cite{uecker2010real, haji2018validation, lim20193d}. 
In response, there have been a number of self-supervised methods proposed, which train on sub-sampled data only \cite{tamir2019unsupervised, huang2019deep, yaman2020self, aggarwal2021ensure}.   

Most existing training methods assume that the measurement noise is small and do not explicitly denoise sampled data. Section \ref{sec:theory_prop} shows theoretically that without explicit denoising the reconstruction quality degrades when the measurement noise increases. This is a particular concern for low SNR measurements, where SNR is the ratio of the signal and noise amplitudes and the SNR is considered ``low" when the measurement noise contributes substantially to the difference between the noisy, sub-sampled data and the ground truth. For instance, the data acquired from low-cost, low-field scanners  is considered low SNR \cite{obungoloch2018design, koonjoo2021boosting, schlemper2020deep}.   

The goal of this paper is to develop a theoretically rigorous, computationally efficient approach for simultaneous self-supervised reconstruction and denoising that performs comparably to fully supervised training. The primary challenge of this goal is that many existing self-supervised denoising methods are not applicable to data that is also sub-sampled \cite{xie2020noise2same}, depend on paired instances of noisy data \cite{lehtinen2018noise2noise}, or are substantially computationally more expensive than fully supervised learning at training time \cite{batson2019noise2self}.  

This paper proposes a modification of Self-Supervised Learning via Data Undersampling (SSDU) \cite{yaman2020self}  that also removes measurement noise, building on the present authors' recent work \cite{millard2022theoretical} on the connection between SSDU and the multiplicative version of the self-supervised denoising method Noisier2Noise \cite{moran2020noisier2noise}. Our method, which we term ``Robust SSDU", combines SSDU with \textit{additive} Noisier2Noise. In brief, Robust SSDU trains a network to map from a further sub-sampled and further noisy version of the training data to the original sub-sampled, noisy data. Then, at inference, a correction is applied to the network output that ensures that the clean (i.e. noise-free) image is recovered in expectation. 

We find that Robust SSDU performs competitively with a fully-supervised benchmark where the network is trained on clean, fully sampled data, despite training on noisy, sub-sampled data only.  We also propose a related method that recovers clean images for the simpler task of when fully sampled, noisy data is available for training, which we term ``Noisier2Full". Both Noisier2Full and Robust SSDU are fully mathematically justified and have minimal additional computational expense compared to standard training. 


The existing method  most similar to Robust SSDU  is Noise2Recon-Self-Supervised (Noise2Recon-SS) \cite{desai2023noise2recon}. The proposed method  Robust SSDU has a number of key difference to Noise2Recon-SS, including a loss weighting and an additive Noisier2Noise correction term at inference that statistically guarantees recovery of the ground truth: see Section \ref{sec:comp} for a detailed comparison. To our knowledge, Robust SSDU is the first method that provably recovers clean images when only noisy, randomly sub-sampled data is available for training. In practice, we find that Robust SSDU offers substantial image quality improvements over Noise2Recon-SS and a two-fold reduction in computational cost at training: see Section \ref{sec:results}. 



\subsection{Notation}

This paper uses notation consistent with \cite{millard2022theoretical}. We use the subscripts $t$ and $s$ to index the training set $\mathcal{T}$ and test set $\mathcal{S}$ respectively. For instance, data in the training and test set are denoted $y_t$ and $y_s$ respectively. Random variables are represented as their instances without indices, and are capitalized if they are vectors. For instance, $y_t, y_s \backsim Y$ for vectors and $ M_{\Omega_t}, M_{\Omega_s} \backsim \Mo$ for matrices. 

We use $Y_0$ to refer to the ground truth, $Y$ to refer to the data, $\Yt$ to refer to the further corrupted data and $\hat{Y}$ to refer to an estimate of the ground truth. We note that sections \ref{sec:add_noise}, \ref{sec:mult_noise} and \ref{sec:theory_prop} onward discuss different recovery tasks, so the definitions of, for instance, the data $Y$ and its instances are section-specific. 

\section{Theory: background}

Image recovery with deep learning is a regression problem, so is centered around the conditional distribution $Y_0|Y$, where $Y_0$ and $Y$ are the random variables associated with the ground truth and data respectively \cite{berk2008statistical}. If ground truth data $y_{0,t} \backsim Y_0 $ are available for training, fully supervised learning can be employed to characterize $Y_0|Y$ directly  \cite{tian2020deep}. This paper focuses on self-supervised learning, which concerns the task of training a network to estimate the ground truth when the training data is $y_t \backsim Y$, so is itself corrupted \cite{lehtinen2018noise2noise, krull2019noise2void, batson2019noise2self, huang2021neighbor2neighbor}. 

The remainder of this section reviews key works from the self-supervised learning literature that form the bases of the methods proposed in this paper. Section \ref{sec:add_noise} presents the case where the data corruption is Gaussian noise, and Section \ref{sec:mult_noise} presents the case where the data  corruption is sub-sampling. 


\subsection{Self-supervised denoising with Noisier2Noise\label{sec:add_noise}}

Denoising with deep learning aims to recover a clean $q$-dimensional vector from noisy data
\begin{align}
y_s = y_{0,s} + n_s,
\end{align}
where $n_s$ is noise and $s \in \mathcal{S}$ indexes the test set. In MRI, noise in k-space is modeled as complex Gaussian with zero mean, $n_s \backsim \mathcal{CN} (0, \Sigma^2_n )$, where $\Sigma^2_n$ is a covariance matrix that can be estimated, for instance, with an empty pre-scan \cite{hansen2015image}. We treat the noise as white, $\Sigma^2_n = \sigma_n^2\eye$, noting that noise with non-trivial covariance can be whitened  by left-multiplying $y_s$ with $\Sigma^{-1}_n$. Other noise distributions are discussed in Section \ref{sec:disc}.


This paper focuses on additive Noisier2Noise \cite{moran2020noisier2noise} because we find that it offers a natural way to extend image reconstruction to low SNR data: see Section \ref{sec:theory_prop}. Noisier2Noise's training procedure consists of corrupting the noisy training data with further noise, and training a network to recover the singly noisy image from the noisier image. 
Concretely, for each $y_t$, further noise is introduced,
\begin{align}
 \yt_t =y_t + \nt_t = y_{0,t} + n_t + \nt_t, \label{eqn:y0_plus_nnt}
\end{align}
where $\nt_t \backsim \mathcal{CN} (0, \alpha^2 \sigma^2_n \eye )$ for a constant $\alpha$. Then, a network $f_\theta$ with parameters $\theta$ is trained to minimize
\begin{align}
	\hat{\theta} = \underset{\theta}{\arg \min} \sum_{t \in \mathcal{T}} \| f_\theta (\yt_t) - y_t \|^2_2. \label{eqn:add_theta_hat}   
\end{align} 
The following result states that a simple transform of the trained network yields the ground truth in expectation despite never seeing the ground truth during training. Here, and throughout this paper, expectations are taken over all random variables. 
\begin{result} \label{res:add_n2n}
Consider the random variables $Y = Y_0 + N$ and $\Yt = Y + \Nt$,  where $N$ and $\Nt$ are zero-mean Gaussian distributed with variances $\sigma_n^2$ and $\alpha^2 \sigma_n^2$ respectively.  Minimizing 
\begin{equation}
	\theta^* = \underset{\theta}{\arg \min} \hspace{0.1cm}
	\Eds[\| f_\theta (\Yt) - Y\|^2_2 | \Yt ] \label{eqn:theta_n2n}
\end{equation} 
yields a network that satisfies
\begin{equation}
	\Eds [Y_0| \Yt] =  \frac{ (1 + \alpha^2)f_{\theta^*}(\Yt) - \Yt }{\alpha^2} . \label{eqn:add_corr}
\end{equation} 
\end{result}
\begin{proof}
	See Section 3.3 of \cite{moran2020noisier2noise}.
\end{proof}
Here, \eqref{eqn:theta_n2n} can be thought of as \eqref{eqn:add_theta_hat} in the limit of an infinite number of samples, and $\hat{\theta}$ as a finite sample approximation of $\theta^*$. Result \ref{res:add_n2n} states that the clean image can be estimated in conditional expectation by employing a correction term based on $\alpha$. It suggests the following procedure for estimating $y_{0,s}$ at inference: corrupt the test data $y_s$ with further noise, $\yt_s = y_s + \nt_s$, apply the trained network to the further noisy data, $f_{\hat{\theta}}(\yt_s)$, and correct the output using the right-hand-side of \eqref{eqn:add_corr}. 





\subsection{Self-supervised reconstruction with SSDU \label{sec:mult_noise}}

This section focuses on the case where the data consists of noise-free, sub-sampled data 
\begin{align}
y_s = M_{\Omega_s} y_{0,s}. 
\end{align}
 Here, $M_{\Omega_s}$ is a sampling mask, a diagonal matrix with $j$th diagonal 1 when $j \in \Omega_s$ and 0 otherwise for sampling set $\Omega_s \subseteq \{1, 2, \ldots, q\}$. 

Self-supervised reconstruction consists of training a network to recover images when only sub-sampled data is available for training: $y_t = M_{\Omega_t} y_{0,t}$  \cite{zeng2021review}. This work focuses on the popular method SSDU \cite{yaman2020self}, which was theoretically justified in \cite{millard2022theoretical} via the multiplicative noise version of Noiser2Noise \cite{moran2020noisier2noise}. In this framework, analogous to the further noise used in \eqref{eqn:y0_plus_nnt}, the training data $y_t$ is \textit{further sub-sampled} by applying a second mask with sampling set $\Lambda_t \subseteq \{1, 2, \ldots, q\}$ to $y_t$, 
\begin{align}
 \yt_t =M_{\Lambda_t} y_t = M_{\Lambda_t \cap \Omega_t} y_{0,t}, \label{eqn:mult_lamb}
\end{align}
where $M_{\Lambda_t \cap \Omega_t} = M_{\Lambda_t} M_{\Omega_t}$. Training  consists of minimizing a loss function on indices in $\Omega_t \setminus \Lambda_t$, such as
\begin{equation}
	\hat{\theta} = \underset{\theta}{\arg \min} \sum_{t \in \mathcal{T}}  \| M_{\Omega_t \setminus \Lambda_t }(  f_\theta (\yt_t) - y_t) \|^2_2, \label{eqn:theta_star_mult}
\end{equation} 
where $M_{\Omega_t \setminus \Lambda_t } =   (\eye -M_{\Lambda_t} ) M_{\Omega_t}$. Although for theoretical ease we state SSDU with an $\ell_2$ loss here, it is known that other losses  are possible \cite{yaman2020self}. 

Let $p_j = \Pds[ j \in \Omega]$ and $\pt_j = \Pds[ j \in \Lambda]$. Assuming that 
\begin{align}
 p_j & > 0 \hspace{0.2cm} \forall \hspace{0.1cm} j,
\label{eqn:mask_req1} 
 \\ \pt_j < 1 & \hspace{0.2cm} \forall \hspace{0.1cm}  \{j : p_j < 1\}, \label{eqn:mask_req2}
\end{align} 
the following result from \cite{millard2022theoretical} proves that SSDU recovers the clean image in expectation.

\begin{result} \label{res:ssdu}
Consider the random variables $Y = \Mo Y_0$ and $\Yt = \Mt Y$. When \eqref{eqn:mask_req1} and \eqref{eqn:mask_req2} hold, minimizing
\begin{equation}
	\theta^* = \underset{\theta}{\arg \min}  \hspace{0.1cm} \Eds[  \| M_{\Omega \setminus \Lambda }(  f_\theta (\Yt) - Y) \|^2_2 | \Yt] \label{eqn:theta_star_mult}
\end{equation} 
yields a network with parameters that satisfies
\begin{equation}
 \Mloc  \Eds[Y_0 | \Yt] = \Mloc f_{\theta^*} (\Yt).  \label{eqn:ssdu_res}
\end{equation}
\end{result} 
\begin{proof}
See Appendix B of \cite{millard2022theoretical}\footnote{Where \cite{millard2022theoretical} uses $\eye - \Mt \Mo$, this uses paper the more compact notation $\Mloc$, where superscript $c$ denotes the complement of a set.}.
\end{proof}
Result \ref{res:ssdu} states that the network correctly estimates $Y_0$ in conditional expectation for indices not in $\Lambda \cap \Omega$. To estimate everywhere in  k-space one can overwrite sampled indices or use a data consistent architecture: see \cite{millard2022theoretical} for details.





\section{Theory: proposed methods \label{sec:theory_prop}}

The reminder of this paper considers the task of training a network to recover images from data that is both noisy \textit{and} sub-sampled:
\begin{align}
 y_s = M_{\Omega_s} (y_{0,s} + n_s). 
\end{align}
It has been stated that when a network reconstructs noisy MRI data with a standard training method, there is a denoising effect \cite{koonjoo2021boosting}. In the following, we motivate the need for methods that explicitly remove noise by showing that the apparent noise removal is in fact a ``pseudo-denoising" effect due to the correct estimation of the ground truth in expectation only for indices in $\Omega^c$. 

Consider the standard approach of training a network to map from noisy, sub-sampled $y_t$ to noisy, fully sampled $y_{0,t} + n_t$. In terms of random variables, training consists of minimizing 
\begin{equation}
	\theta^* = \underset{\theta}{\arg \min} \hspace{0.1cm} 
	\Eds [ \| f_\theta (Y) - (Y_0 + N)\|^2_2 | Y],
\end{equation} 
which gives a network that satisfies
\begin{equation}
f_{\theta^*} (Y) = \Eds [Y_0 + N | Y]. \label{eqn:motiv_f}
\end{equation}
Eqn. \eqref{eqn:motiv_f} does not hold for completely arbitrary network architecture. The conditions on $f_\theta$ (which are also required for Results \ref{res:add_n2n} and \ref{res:ssdu}) are detailed in Section II-A of \cite{millard2022theoretical}. In brief, the Jacobian matrix $J$ with entries $J_{ij} = \partial f_{\theta} (Y)_j / \partial \theta_i$ must have maximally linearly independent rows, which is expected for well-constructed architectures when the number of parameters exceeds $q$. Throughout the remainder of this paper, we assume that $f_\theta$ satisfies this condition. We also assume that the optimizer is not stuck in a poor local minimum so that the network is a good approximation of \eqref{eqn:motiv_f} in practice.

It is instructive to examine how $\Eds [Y_0 + N | Y]$ depends on the sampling mask $\Omega$. Firstly, for $j \notin \Omega$, 
\begin{align}
\Eds [Y_{0,j}  + N_j | Y, j \notin \Omega] &= \Eds [Y_{0,j} | Y]  + \Eds [N_j]   \nonumber \\ &= \Eds [Y_{0,j} | Y], \label{eqn:motiv_notin}
\end{align}
where we have used the independence of $N_j$ from $Y$ when $j \notin \Omega$ and $\Eds [N_j] = 0$ by assumption. For the alternative   $j \in \Omega$,  
\begin{align} 
\Eds [Y_{0,j} + N_j | Y, j \in \Omega ] = \Eds [Y_{j}| Y] = Y_j
\end{align}
where $Y_{0,j} + N_j = Y_j$ for $j \in \Omega$ has been used.
The trained network therefore satisfies
\begin{align}
f_{\theta^*}(Y) = \Eds [Y_0 + N | Y] = M_{\Omega^c} \Eds [Y_0 | Y]  +  \Mo Y. \label{eqn:18}
\end{align}
Therefore the network targets the noise-free $Y_{0}$ in regions in $\Omega^c$ but recovers the noisy $Y$ otherwise. As there is less total measurement noise present than $Y_0 + N$, this gives the impression of noise removal, however, we emphasize that the network does not remove the noise in $Y$. Since the term ``denoising" typically refers to the removal of noise \textit{from the input data}, we use the term ``pseudo-denoising" to refer to the behavior stated in \eqref{eqn:18}. Other than the conditions on $f_\theta$ described above, this result is agnostic to the network architecture, so includes ``unrolled" approaches that may have a regularization parameter which is designed to trade off the model and consistency with the data.   

 We refer to this method described in this section as ``Supervised w/o denoising" throughout this paper. 
In the following we propose methods that explicitly recover $Y_0$ in conditional expectation from noisy, sub-sampled inputs in two cases: A) the training data is noisy and fully sampled; B) the training data is noisy and sub-sampled. For tasks A and B we propose ``Noisier2Full" and ``Robust SSDU" respectively.

\subsection{Noisier2Full for fully sampled, noisy training data \label{sec:denoi_prop}}

This section proposes Noisier2Full, which  extends additive Noisier2Noise to reconstruction tasks for noisy, fully sampled training data. Based on \eqref{eqn:y0_plus_nnt}, we propose corrupting the measurements $y_t$ with further noise on the sampled indices,
\begin{align}
  \yt_t = y_{t} + M_{\Omega_t} \nt_t.
\end{align}
Then we minimize the loss between $\yt_t$ and the noisy, fully sampled training data $y_{0,t} + n_t$. In terms of random variables,
\begin{equation}
	\theta^* = \underset{\theta}{\arg \min} \hspace{0.1cm} 
	\Eds [ \| f_\theta (\Yt) - (Y_0 + N)\|^2_2 | \Yt]. 
	\label{eqn:theta_n2n_ssdu}
\end{equation} 
Minimizing the $\ell_2$ norm gives a network that satisfies
\begin{align}
	f_{\theta^*} (\Yt) = \Eds [Y_0 + N | \Yt], 
\end{align}
which is recognizable as \eqref{eqn:motiv_f} with $Y$ replaced by $\Yt$. Similarly to \eqref{eqn:motiv_notin}, $N_j$ is independent of $\Yt$ when $j \notin \Omega $, so the ground truth is estimated in such regions:
\begin{equation}
\Eds [Y_{0, j} | \Yt, j \notin \Omega] = \Eds [Y_{0, j} | \Yt].
\end{equation}
However, crucially, the expectation is conditional on $\Yt$, not $Y$, so the additive Noisier2Noise correction stated in Result 1 is applicable when $j \in \Omega$:
\begin{align}
	\Eds [Y_{0,j}| \Yt, j \in \Omega] = \frac{ (1 + \alpha^2)f_{\theta^*}(\Yt)_j - \Yt_j }{\alpha^2}
\end{align}
Although Result \ref{res:add_n2n} is not specifically constructed for sub-sampled data, it is applicable here because it is an entry-wise statistical relationship, so can be applied to each index that has the proper noise statistics. Therefore $Y_0$ can be estimated with
\begin{equation}
	\Eds [Y_0| \Yt] = \Mo \left ( \frac{ (1 + \alpha^2)f_{\theta^*}(\Yt) - \Yt }{\alpha^2} \right ) + M_{\Omega^c} f_{\theta^*}(\Yt). \label{eqn:yhat}
\end{equation}
In summary, Noisier2Full recovers $Y_0$ in conditional expectation by introducing further noise to the sampled indices during training, and correcting those indices at inference via additive Noisier2Noise. In the subsequent section, we show how this approach can be extended to the more challenging case where the training data is also sub-sampled. 


\subsection{Robust SSDU for sub-sampled, noisy training data \label{sec:denoi_and_dealias}}

\begin{figure}[t]
\centering
	\includegraphics[width=0.5\textwidth]{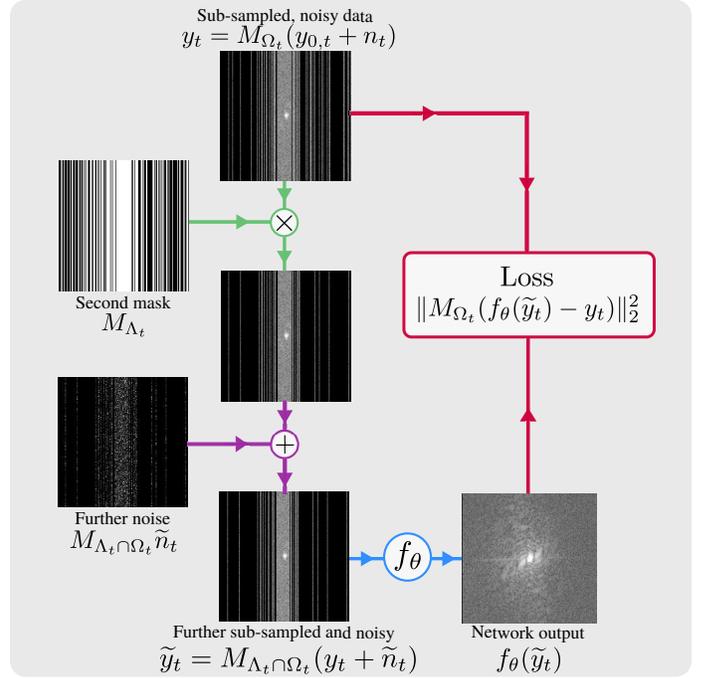}
	\caption{The proposed self-supervised reconstruction and denoising method Robust SSDU, which extends the training procedure illustrated in Fig. 1 of \cite{millard2022theoretical} to low SNR data. The sub-sampled, noisy training data $y_t$ is further sub-sampled by a mask $M_{\Lambda_t}$ and corrupted by further noise $\nt_t$, yielding $\yt_t$. The loss is computed between $y_t$ and $f_\theta(\yt_t)$ on $\Omega_t$. \label{fig:schematic}}
\end{figure}

This section proposes Robust SSDU, which recovers clean images in conditional expectation when the training data is both noisy and sub-sampled. Robust SSDU combines the approaches from Sections \ref{sec:add_noise} and \ref{sec:mult_noise} to simultaneously reconstruct and denoise the data: see Fig. \ref{fig:schematic} for a schematic. We propose combining \eqref{eqn:y0_plus_nnt} and \eqref{eqn:mult_lamb} to form a vector that is further sub-sampled \textit{and} additionally noisy,   
\begin{align}
  \yt_t = M_{\Lambda_t\cap \Omega_t} (y_{t} + \nt_t). 
\end{align}
Recall that SSDU employs $M_{\Omega \setminus \Lambda }$ in the loss, which yields a network that estimates indices in $(\Lambda \cap \Omega)^c$: see Result \ref{res:ssdu}. For Robust SSDU we replace $M_{\Omega\setminus \Lambda }$ with $M_{\Omega}$, so that the loss is 
\begin{equation}
	\hat{\theta} = \underset{\theta}{\arg \min} 
	\sum_{t \in \mathcal{T}} \| M_{\Omega_t}  (f_\theta (\yt_t) - y_t)\|^2_2. 
\end{equation} 
In the following we show that this change leads to estimation everywhere in k-space, not just indices in $(\Lambda \cap \Omega)^c$. 

%

\begin{claim} \label{clm:robust_ssdu}
Consider the random variables $Y = \Mo (Y_0 + N)$ and $ \Yt = M_{\Lambda \cap \Omega}  (Y + \Nt)$,  where $N$ and $\Nt$ are zero-mean Gaussian distributed with variances $\sigma_n^2$ and $\alpha^2 \sigma_n^2$ respectively. When \eqref{eqn:mask_req1} and \eqref{eqn:mask_req2} hold, minimizing 
\begin{align}
	\theta^*  = \underset{\theta}{\arg \min} \hspace{0.1cm} \Eds[ \| \Mo (f_{\theta} (\Yt) - Y) \|_2^2 | \Yt ] \label{eqn:LSSDU}
\end{align} 
	yields a network with parameters that  satisfies
\begin{align}
 f_{\theta^*}(\Yt) = \Eds [Y_0 + N | \Yt].  
\end{align}
\end{claim}
\begin{proof}
See appendix \ref{app:mo}.
\end{proof}
%

The differences between \eqref{eqn:LSSDU} and the standard SSDU loss \eqref{eqn:theta_star_mult} are the change from $M_{\Omega \setminus \Lambda }$ to $\Mo$ and the inclusion of noise in the data $Y$. Intuitively, since $\Mo = M_{\Omega \setminus \Lambda } +  \Mlo$, the mask change extends \eqref{eqn:theta_star_mult} to include entries in $\Mlo$. Therefore, at inference, the network learns to map to entries in $\Mloc$, as stated in Result \ref{res:ssdu}, \textit{and} $ \Mlo$, which comes from the additional indices in the loss. In other words, it learns to map to everywhere in k-space. The inclusion of noise in the target simply implies that the network will learn to map to the noisy  $Y_0 + N$, as in \eqref{eqn:motiv_f}. 

At inference, we can use a similar approach to Section \ref{sec:denoi_prop}, applying the additive Noisier2Noise correction on indices sampled in $\Yt$. Since the indices sampled in $\Yt$ are $\Lambda \cap \Omega$, the clean image $Y_0$ is estimable with
\begin{multline}
	\Eds [Y_0| \Yt]  = \Mlo \left ( \frac{ (1 + \alpha^2)f_{\theta^*}(\Yt) - \Yt }{\alpha^2} \right ) \\ + \Mloc f_{\theta^*}(\Yt). \label{eqn:yhat_n2n_ssdu}
\end{multline}
Roughly speaking, Robust SSDU can be thought of as a generalization of Noisier2Full to sub-sampled training data. Specifically, Robust SSDU is mathematically equivalent to Noisier2Full when $\Omega = \{1, 2, \ldots, q\}$ and there is the change of notation $\Lambda \rightarrow \Omega$. More broadly, Robust SSDU can be interpreted as the simultaneous application of additive and multiplicative  Noisier2Noise \cite{moran2020noisier2noise, millard2022theoretical}.  

\subsection{Loss weighting of Noisier2Full and Robust SSDU}

For Noisier2Full and Robust SSDU, the task at training and inference is not identical: at training the network maps from $\Yt$ to $Y_0 + N$ or $\Mo(Y_0 + N)$, while at inference it maps from $\Yt$ to $Y_0$ via the $\alpha$-based correction term. Taking a similar approach to \cite{wang2023k, wiedemann2023deep}, this section describes how this can be compensated for by modifying the loss function in such a way that its gradient equals the gradient of the target loss in conditional expectation.

\begin{claim} \label{clm:Nr2N} Consider the random variables $Y = \Mo (Y_0 + N)$ and $\Yt = Y + \Mo \Nt$, where $N$ and $\Nt$ are zero-mean Gaussian distributed with variances $\sigma_n^2$ and $\alpha^2 \sigma_n^2$ respectively. Define 
\begin{align}
\Yhn = \Mo \left ( \frac{ (1 + \alpha^2)f_{\theta}(\Yt) - \Yt }{\alpha^2} \right ) + M_{\Omega^c} f_{\theta}(\Yt) \label{eqn:n2f_yhat}
\end{align}
where $f_{\theta}$ is an arbitrary function. Then
\begin{multline}
	 \dt \Eds \left[ \left\|\hat{Y}_{Nr2F} - Y_0 \right \|^2_2  | \Yt \right] \\ = \dt \Eds \left[ \left\|W_{\Omega} (f_\theta ( \Yt) - Y_0 - N) \right \|^2_2  | \Yt \right]. \label{eqn:nr2f_clm}
\end{multline}
where 
\begin{equation}
W_{\Omega} = \A \Mo  + \Moc.
\end{equation}
\end{claim}
\begin{proof}
See appendix \ref{app:Nr2N}. 
\end{proof}

We therefore suggest replacing the Noisier2Full loss stated in \eqref{eqn:theta_n2n_ssdu} with the right-hand-side of \eqref{eqn:nr2f_clm}, which   increases the weight of the indices in $\Omega$. Intuitively, it uses the ratio of noise removed at training, which has variance $\mathrm{Var}(\Nt) = \alpha^2 \sigma_n^2$, and the noise removed at inference, which has variance $\mathrm{Var}(N + \Nt) = (1 + \alpha^2) \sigma_n^2$, to compensate for the difference between the task at training and inference. The following result concerns the analogous expression for Robust SSDU.

\begin{claim} \label{clm:Nr2N} Consider the random variables $Y = \Mo (Y_0 + N)$ and $\Yt = \Mlo (Y +  \Nt)$, where $N$ and $\Nt$ are zero-mean Gaussian distributed with variances $\sigma_n^2$ and $\alpha^2 \sigma_n^2$ respectively. Define 
\begin{multline}
\Yhr = \Mlo \left ( \frac{ (1 + \alpha^2)f_{\theta}(\Yt) - \Yt }{\alpha^2} \right ) \\ + \Mloc f_{\theta^*}(\Yt)
\end{multline}
where $f_{\theta}$ is an arbitrary function. Then
\begin{multline}
	 \Eds \left[ \left\|\Yhr - Y_0 \right \|^2_2  | \Yt  \right] \\ = \dt \Eds \left[ \left\| W_{\Omega, \Lambda} M_\Omega (f_\theta ( \Yt) - Y) \right \|^2_2   | \Yt \right ] 
\end{multline}
where 
\begin{equation}
W_{\Omega, \Lambda} = \A \Mlo  + \mathcal{P}^{\frac{1}{2}}  \Moml 
\end{equation}
and $\mathcal{P} = \Eds[ \Moml] ^{-1}  \Eds[ \Mloc]$. \label{eqn:p_def}

\end{claim}
\begin{proof}
See appendix \ref{app:rssdu}. 
\end{proof}
The $\Mlo$ coefficient has a similar role to the $M_\Omega$ coefficient in  \eqref{eqn:nr2f_clm}. The $\Moml$ coefficient compensates for the variable density of $\Omega$ and $\Lambda$, and was first proposed in \cite{millard2022theoretical}, where it was shown to improve the reconstruction quality and robustness to the distribution of $\Lambda$ for standard SSDU without denoising.\footnote{Where \cite{millard2022theoretical} uses $(\eye - K)^{-1}$, this paper uses the more compact $\Pc$.} 

The weightings can be thought of as entry-wise modifications of the learning rate \cite{millard2022theoretical}. Neither weighting matrices change $\theta^*$, so the proofs of Noisier2Full and Robust SSDU from Sections \ref{sec:denoi_prop} and \ref{sec:denoi_and_dealias} hold. Rather, the role of the weights is to improve the finite-sample case in practice, where   $\theta^*$ is estimated with $\hat{\theta}$: see Section \ref{sec:results} for an empirical evaluation. Throughout the remainder of this paper, ``Noisier2Full" and ``Robust SSDU" refer to the versions with the loss weightings proposed in this section and versions without such weightings are explicitly referred to as ``Unweighted Noisier2Full" and ``Unweighted  Robust SSDU".


\section{Materials and methods}
\label{sec:exp_meth}

\subsection{Description of data}

We primarily used the multi-coil brain data from the publicly available fastMRI dataset \cite{zbontar2018fastmri}\footnote{available from \url{https://fastmri.med.nyu.edu}}. 
We only used data that had 16 coils, so that the training, validation and test sets contained 2004, 320 and 224 slices respectively. The slices were normalized so that the cropped RSS estimate had maximum 1. Here, the cropped RSS is defined as $ Z ((\sum_c^{N_c} | F^H {y}_c|^2)^\frac{1}{2})$, where the subscript $c$ refers to all entries on the $c$th coil, $F$ is the discrete Fourier transform, $N_c$ is the number of coils and $Z$ is an operator that crops to a central $320 \times 320$ region. RSS images were used for normalization and visualization only; otherwise, the raw complex multi-coil k-space data was used. We retrospectively sub-sampled column-wise with the central 10 lines fully sampled and randomly drawn with polynomial density otherwise, with the probability density scaled to achieve a desired acceleration factor $R_\Omega=q/\sum_j p_j$. For $R_\Omega = 4$ and $\sigma_n = 0.04$, we also trained the methods on 2D Bernoulli sampling, where the sampling was random and independent, also with polynomial variable density. For each case, the distribution of  $M_\Lambda$ was the same type as the first \cite{millard2022theoretical}.  The data was treated as noise-free, and we generated white, complex Gaussian measurement noise with standard deviation $\sigma_n$ to simulate noisy conditions.

We also tested the methods' performance on the 0.3T dataset M4Raw \cite{lyu2023m4raw}. For this dataset, which is prospectively low SNR, no further noise was added. Rather, the noise covariance matrix was estimated using the fully-sampled image via a $30 \times  30$ square of background from each corner and the data was whitened by left-multiplying with the inverse covariance matrix, so that all data had noise standard deviation 1. The same column-wise sub-sampling was used as described above for the fastMRI data.  Although more realistic that the simulated noise setting of fastMRI, for M4Raw we have no ``ground truth", so it was only possible to evaluate the methods' performance qualitatively.

 An implementation of our method in PyTorch is available on GitHub\footnote{\url{https://github.com/charlesmillard/robust_ssdu}}.


\subsection{Comment on the proposed methods in practice}

The theoretical guarantees for Noisier2Full and Robust SSDU use the further noisy, possibly further sub-sampled $\yt_s$ as the input to the network at inference. 
In practice, as suggested in the original Noisier2Noise \cite{moran2020noisier2noise} and SSDU  \cite{yaman2020self} papers, we used $y_s$ as the input to the network at inference, so that the estimate
\begin{equation}
	\hat{y}_s = M_{\Omega_s}  \left( \frac{(1+\alpha^2)f_{\hat{\theta}} (y_s) - y_s}{\alpha^2} \right) + M_{\Omega_s^c}  f_{\hat{\theta}} (y_s)  
\end{equation} 
is used in place of \eqref{eqn:yhat} and \eqref{eqn:yhat_n2n_ssdu}. Although this deviates from strict theory, and is not guaranteed to be correct in conditional expectation, we have found that it achieves better  reconstruction performance in practice: see \cite{moran2020noisier2noise} and \cite{millard2022theoretical} for a detailed empirical evaluation. All subsequent results for the proposed methods use this estimate at inference.

\subsection{Comparative training methods \label{sec:comp}}


The training methods evaluated in this paper are summarized in Table \ref{tab:meths}.

For noise-free, fully sampled training data, fully supervised training can be employed, where the loss is computed between the output of the network $f_\theta (y_t)$ and the noise-free, fully sampled target $y_{0,t}$: see Table \ref{tab:meths}. Although it is possible in principle to have higher SNR data at training than at inference by acquiring multiple averages \cite{lyu2023m4raw}, such datasets would require an extended acquisition time and are rare in practice. Nonetheless, training a network on this type of data via simulation is instructive as a best-case target. This method is referred to as the ``fully-supervised benchmark" throughout this paper.  

For noisy, fully sampled training data, we employed three training methods: Unweighted Noisier2Full, Noisier2Full and the standard approach Supervised w/o denoising, as described in Section \ref{sec:theory_prop}. We did not compare with Noise2Inverse \cite{hendriksen2020noise2inverse} as it was designed for learned denoising but fixed reconstruction operators.

For the more challenging scenario where noisy, sub-sampled training data is available, we compared Robust SSDU with the original version of SSDU, which reconstructs sub-sampled data but does not denoise. We refer to this as ``Standard SSDU". We also compared with Noise2Recon-SS \cite{desai2023noise2recon}, which, like Robust SSDU, includes adding further noise to the sub-sampled data. However, Noise2Recon-SS has a number of key differences to the method proposed in this paper. With an $\ell_2$ k-space loss, training with Noise2Recon-SS consists of minimizing  
\begin{multline}
 \hat{\theta} =	\underset{\theta}{\arg \min} \sum_{t \in \mathcal{T}} \| M_{\Omega_t \setminus \Lambda_t} (f_\theta(M_{\Lambda_t} y_t) - y_t) \|^2_2  \\ + \lambda \| f_\theta( y_t + M_{\Omega_t}\nt_t) - f_\theta(M_{\Lambda_t} y_t)\|^2_2, \label{eqn:n2r_loss}
\end{multline}
where $\lambda$ is a hand-selected weighting. We used $\lambda = 1$ throughout. The $\ell_2$ loss in k-space was used so that it could be fairly compared to the other methods in this paper, but we note that \cite{desai2023noise2recon} used image-domain losses. The first term is based on SSDU, and the second ensures that $f_\theta( y_t + M_{\Omega_t}\nt_t)$ and $f_\theta(M_{\Lambda_t} y_t)$ yield similar outputs, so that the method is in a sense robust to $\nt_t$. At inference, Noise2Recon-SS uses $\hat{y}_s = f_{\hat{\theta}}(y_s)$; there is no correction term. We emphasize that, unlike the proposed Robust SSDU, there is no theoretical evidence that Noise2Recon-SS recovers the clean image in expectation. 

In \cite{koonjoo2021boosting}, an untrained denoising algorithm was appended to a reconstruction network.  To test this, we denoised the RSS output of Supervised w/o denoising and Standard SSDU with the popular BM3D algorithm \cite{Dabov2007}, which is designed for Gaussian noise. Although the measurement noise is Gaussian, the reconstruction error of the RSS image is not Gaussian in general \cite{Milla2020}. Therefore, unlike the proposed methods, BM3D does not accurately model the noise characteristics  \cite{Virtue2017}. Nonetheless, we found that these methods performed reasonably well in practice.



                                                                      
\begin{figure}[t]
\centering
	\includegraphics[width=0.5\textwidth]{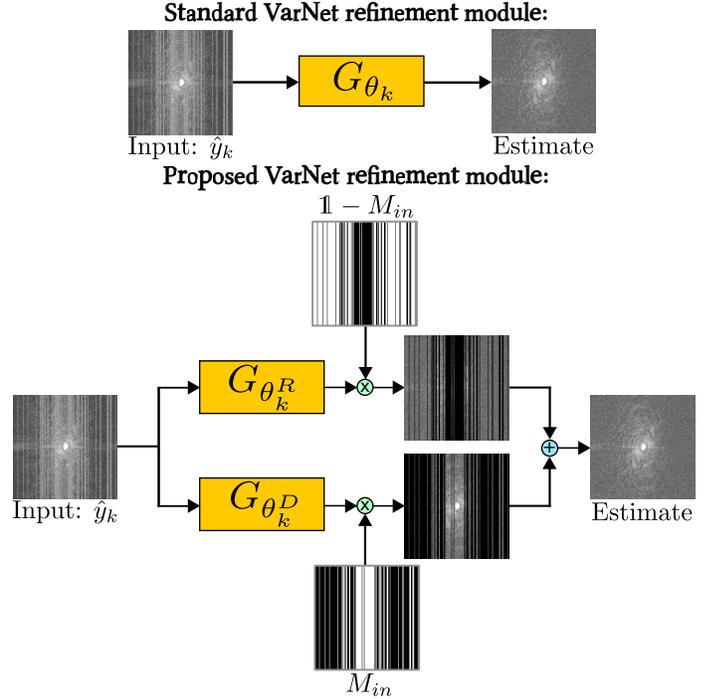}
	\caption{The refinement module for the proposed architecture Denoising VarNet, which trains two networks in parallel, removing noise and aliasing separately. \label{fig:prop_refine}}
\end{figure}

\newcolumntype{P}[1]{>{\centering\arraybackslash}p{#1}}
\begin{table*}[th]
\centering
\scriptsize
\begin{tabular}{c|c|P{60mm}|c}
\textsc{Name}         & \textsc{Training data} & \textsc{Loss}                                             & \textsc{Estimate at inference}  \\\hhline{=|=|=|=}
Fully-supervised benchmark & $y_{0,t}$ & $\sum_{t \in \mathcal{T}} \| f_\theta(y_t) - y_{0,t} \|^2_2$ & $f_{\theta} (y_s)$                                               \\[0pt] \hhline{=|=|=|=}
Supervised w/o denoising   & $y_{0,t} + n_t $ & $\sum_{t \in \mathcal{T}} \| f_\theta(y_t) - (y_{0,t} + n_t) \|^2_2$                                            & $f_{\hat{\theta}} (y_s)$       \\[0pt] \hline 
Supervised with BM3D denoising   & $y_{0,t} + n_t $ & $\sum_{t \in \mathcal{T}} \| f_\theta(y_t) - (y_{0,t} + n_t) \|^2_2$                                            & BM3D$(f_{\hat{\theta}} (y_s))$                                                                                                                                                                                                                 \\[0pt] \hline 
Noisier2Full$^*$  & $y_{0,t} + n_t $ &  $\sum_{t \in \mathcal{T}} \| W_{\Omega_t}(f_\theta(y_{t} + M_{\Omega_t} \nt_t) - (y_{0,t} + n_t) )\|^2_2 $ & $M_{\Omega_s}  \left( \frac{(1+\alpha^2)f_{\hat{\theta}} (y_s) - y_s}{\alpha^2} \right) + M_{\Omega_s^c}  f_{\hat{\theta}} (y_s)$                                                                                                      \\[0pt] \hhline{=|=|=|=}
Standard SSDU  & $y_t$ & $\sum_{t \in \mathcal{T}} \| M_{\Omega_t \setminus \Lambda_t} (f_\theta(M_{\Lambda_t} y_t) - y_t) \|^2_2$ & $f_{\hat{\theta}} (y_s)$                                                                                                                                                                                                           \\[0pt] \hline
SSDU  with BM3D & $y_t$ & $\sum_{t \in \mathcal{T}} \| M_{\Omega_t \setminus \Lambda_t} (f_\theta(M_{\Lambda_t} y_t) - y_t) \|^2_2$ & BM3D$(f_{\hat{\theta}} (y_s))$                                                                                                                                                                                                           \\[0pt] \hline
\begin{tabular}{c} Noise2Recon-SS \end{tabular}& \begin{tabular}{c}$y_t$ \end{tabular} & $\sum_{t \in \mathcal{T}} \| M_{\Omega_t \setminus \Lambda_t} (f_\theta(M_{\Lambda_t} y_t) - y_t) \|^2_2$ \newline $+ \lambda \| f_\theta( y_t + M_{\Omega_t}\nt_t) - f_\theta(M_{\Lambda_t} y_t)\|^2_2$ & \begin{tabular}{c} $f_{\hat{\theta}} (y_s)$\end{tabular}                                                                                                       \\[0pt] \hline
Robust SSDU$^*$  & $y_t$ & $\sum_{t \in \mathcal{T}} \| W_{\Omega_t, \Lambda_t} M_{\Omega_t}( f_\theta(M_{\Lambda_t\cap \Omega_t} (y_{t} + \nt_t)) - y_t) \|^2_2$ & $M_{\Omega_s}  \left( \frac{(1+\alpha^2)f_{\hat{\theta}} (y_s) - y_s}{\alpha^2} \right) + M_{\Omega_s^c}  f_{\hat{\theta}} (y_s)$                                                                                                     
\end{tabular}
\caption{The training methods evaluated in this paper, where $y_t = M_{\Omega_t} ( y_{0,t} + n_t )$ and the asterisk denotes the proposed methods. Here, and throughout this paper, the subscripts $t$ and $s$ index the training and test sets respectively. The function BM3D$(\cdot)$ is defined here to include an RSS transform, so that the denoiser acts on the RSS image. The double lines are used to separate types of data available for training. The unweighted variants of Noisier2Full and Robust SSDU, which are not stated here for brevity, are equivalent to the weighted versions with $W_{\Omega_t}= \eye$ and $W_{\Omega_t, \Lambda_t} = \eye$. \label{tab:meths}}
\end{table*}

\subsection{Network architecture}

For all methods considered in this paper, the function $f_\theta$ is defined to be k-space to k-space, but is otherwise agnostic to the network architecture. Architectures can include inverse Fourier transforms, so convolutional layers may be applied in the image domain. We emphasize that the experiments in this paper are designed to compare the performance of the training \textit{method}, not to provide a comprehensive evaluation of possible architectures, which is a somewhat orthogonal goal. 

We employed a network architecture based on the Variational Network (VarNet) \cite{hammernik2018learning, sriram2020end}, which is available as part of the fastMRI package \cite{zbontar2018fastmri}. VarNet consists of a coil sensitivity map estimation module followed by a series of  ``cascades".  The k-space estimate at the $k$th cascade takes the form
\begin{equation}
	\hat{y}_{k+1} = \hat{y}_k - \eta_k M_{in} (\hat{y}_k - y_{in}) + G_{\theta_k} ( \hat{y}_{k})
\end{equation}
where $y_{in}$ and $M_{in}$ are the input k-space and sampling mask respectively and the $t$ or $s$ index has been dropped for legibility. We use the generic subscript $in$ here because the input is not the same for every method: for instance, fully-supervised and Noisier2Full have $M_{in} = M_{\Omega_t}$ and $M_{in} = M_{\Lambda_t \cap \Omega_t}$ respectively. Here, $\eta_k$ is a trainable parameter and $G_{\theta_k} ( \hat{y}_{k})$ is a neural network with cascade-dependent parameters $\theta_k$, referred to as a ``refinement module", which was an image-domain U-net \cite{ronneberger2015u} with real weights in \cite{hammernik2018learning, sriram2020end}. 

VarNet was originally constructed for reconstruction only, without explicit denoising. For joint reconstruction and denoising, we propose partitioning $G_{\theta_k} ( \hat{y}_{k})$ into two functions,
\begin{align}
	G_{\theta_k} ( \hat{y}_{k}) = M_{in} G_{\theta^D_k} ( \hat{y}_{k}) + (\eye - M_{in})G_{\theta^R_k} ( \hat{y}_{k}).
\end{align}
This refinement module is illustrated in Fig. \ref{fig:prop_refine}. We refer to the architecture with the proposed refinement module as ``Denoising VarNet" throughout this paper.  We used a U-net \cite{ronneberger2015u} for both $G_{\theta^D_k} ( \hat{y}_{k})$ and $G_{\theta^R_k} ( \hat{y}_{k})$, although we note that in general these functions need not be the same. We used 5 cascades, giving a network with $2.5 \times 10^7$ parameters.

\subsection{Training details}         
                          
 We used the Adam optimizer \cite{kingma2014adam} and trained for 100 epochs with learning rate $10^{-3}$.  The $\Omega_t$ and $n_t$ were fixed but the $\Lambda_t$ and $\nt_t$ were re-generated once per epoch \cite{9433924}, which we found considerably reduced susceptibility to overfitting. As  in \cite{millard2022theoretical}, we used the same distribution of $\Lambda_t$ as $\Omega_t$ but with parameters selected to give a sub-sampling factor of $R_\Lambda = q/\sum_j \pt_j =  2$ unless otherwise stated. The choice of $\alpha$ is discussed in Section \ref{sec:alpha_rob}. Unless otherwise stated, the training methods were evaluated on data generated with $\sigma_n \in \{0.02, 0.04, 0.06, 0.08\}$ and $R_\Omega \in \{4, 8\}$. We note that the noise standard deviation, not the SNR, was fixed and that for each training method, $\sigma_n$ and $R_\Omega$, we trained a separate network from scratch. 

\subsection{Performance metrics}

Since each of the methods were trained using a squared error loss in k-space, we primarily focused on the k-space normalized mean squared error (NMSE) over the test set, defined as $\frac{1}{|\mathcal{S}|} \sum_{s \in \mathcal{S}} \| \hat{y}_s - y_{0,s}\|^2_2  / \|y_{0,s}\|^2_2$ where $\hat{y}_s$ is an estimate of k-space. Since the score is in k-space, it was not possible to compute the NMSE of methods that employed BM3D, which acts on the magnitude image so does not retain the complex phase. 

We also computed the mean Structural Similarity (SSIM) \cite{Wang2004} on the RSS images. We emphasize that the networks were \textit{not} trained to minimize the SSIM directly, so such scores are somewhat incidental to the primary NMSE results and not necessarily fundamental to the method.

\section{Results \label{sec:results}}



\begin{figure}[]
\centering
	\includegraphics[width=0.35\textwidth]{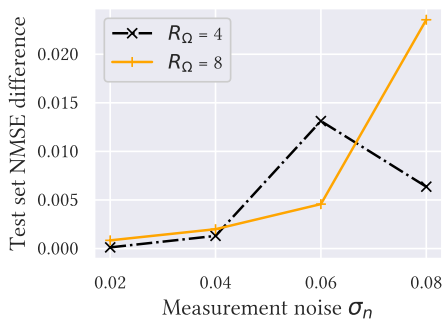}
	\caption{The difference between the test set loss of Standard VarNet and the proposed Denoising VarNet for the benchmark training method. All differences are positive, showing that Denoising VarNet outperforms Standard VarNet, especially for large $\sigma_n$. \label{fig:arch}}
\end{figure}

\begin{figure}[]
\centering
	\includegraphics[width=0.45\textwidth]{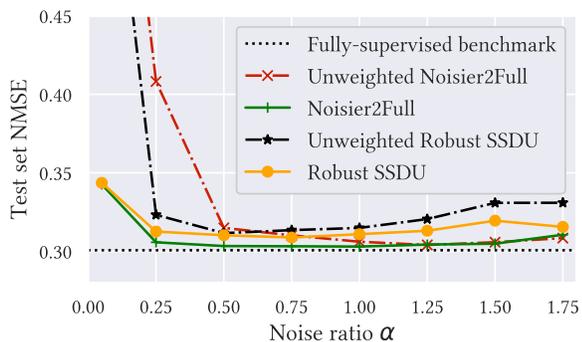}
	\caption{The robustness to  $\alpha$  of Noisier2Full, Robust SSDU and their weighted versions at $R_\Omega = 8$ and $\sigma_n = 0.06$. The performance of the fully-supervised benchmark, which does not depend on $\alpha$, is also shown. The weighted versions are substantially more robust, especially for small $\alpha$: at $\alpha = 0.05$, the values of unweighted Noisier2Full and Robust SSDU, which excluded for visualization,  are 0.70 and 0.62 respectively. \label{fig:alpha_line}}
\end{figure}

\renewcommand{\arraystretch}{1.5}
\begin{table*}[t]
\centering
\fontsize{5.6pt}{8pt}\selectfont
\begin{tabular}{c||c|c|c|c||c|c|c|c}
\multicolumn{1}{c}{} & \multicolumn{4}{c}{\textsc{Acceleration factor} $R_\Omega=4$ } & \multicolumn{4}{c}{\textsc{Acceleration factor} $R_\Omega=8$ } \\ 
      & $\sigma_n=0.02$ & $\sigma_n = 0.04$ & $\sigma_n = 0.06$ & $\sigma_n = 0.08$ & $\sigma_n=0.02$ & $\sigma_n = 0.04$ & $\sigma_n = 0.06$ & $\sigma_n = 0.08$  \\\hhline{=#=|=|=|=#=|=|=|=}
Noisy \& sub-sampled    & $0.210\pm0.01$ & $0.434\pm0.01$ & $0.809\pm0.02$ & $1.333\pm0.02$ & $0.207\pm0.02$ & $0.337\pm0.02$ & $0.554\pm0.02$ & $0.857\pm0.02$    \\[0pt] \hline
Fully-supervised benchmark & $\mathbf{0.167\pm0.02}$ & $\mathbf{0.313\pm0.02}$ & $\mathbf{0.537\pm0.02}$ & $\mathbf{0.850\pm0.02}$ & $\mathbf{0.160\pm0.02}$ & $\mathbf{0.217\pm0.02}$ & $\mathbf{0.301\pm0.02}$ & $\mathbf{0.414\pm0.02}$\\[0pt] \hhline{=#=|=|=|=#=|=|=|=}
Supervised w/o denoising & ${0.187\pm0.02}$ & ${0.412\pm0.02}$ & ${0.788\pm0.02}$ & ${1.314\pm0.02}$ & ${0.178\pm0.02}$ & ${0.310\pm0.02}$ & ${0.527\pm0.02}$ & ${0.833\pm0.02}$\\[0pt] \hline
Unweighted Noisier2Full* & ${0.170\pm0.02}$ & ${0.319\pm0.02}$ & ${0.548\pm0.02}$ & ${0.870\pm0.02}$ & ${0.164\pm0.02}$ & ${0.223\pm0.02}$ & ${0.315\pm0.02}$ & ${0.441\pm0.02}$\\[0pt] \hline
Noisier2Full* & $\mathbf{0.169\pm0.02}$ & $\mathbf{0.312\pm0.02}$ & $\mathbf{0.538\pm0.02}$ & $\mathbf{0.853\pm0.02}$ & $\mathbf{0.162\pm0.02}$ & $\mathbf{0.220\pm0.02}$ & $\mathbf{0.305\pm0.02}$ & $\mathbf{0.422\pm0.02}$\\[0pt] \hhline{=#=|=|=|=#=|=|=|=}
Standard SSDU & ${0.188\pm0.01}$ & ${0.413\pm0.01}$ & ${0.787\pm0.01}$ & ${1.310\pm0.01}$ & ${0.180\pm0.01}$ & ${0.312\pm0.01}$ & ${0.531\pm0.01}$ & ${0.838\pm0.01}$\\[0pt] \hline
Noise2Recon-SS & ${0.180\pm0.02}$ & ${0.377\pm0.02}$ & ${0.623\pm0.02}$ & ${0.975\pm0.02}$ & ${0.173\pm0.02}$ & ${0.260\pm0.02}$ & ${0.452\pm0.02}$ & ${0.691\pm0.02}$\\[0pt] \hline
Unweighted Robust SSDU* & ${0.170\pm0.02}$ & $\mathbf{0.314\pm0.02}$ & ${0.548\pm0.02}$ & ${0.863\pm0.02}$ & $\mathbf{0.162\pm0.02}$ & $\mathbf{0.222\pm0.02}$ & $\mathbf{0.309\pm0.02}$ & ${0.424\pm0.02}$\\[0pt] \hline
Robust SSDU* & $\mathbf{0.169\pm0.02}$ & ${0.315\pm0.02}$ & $\mathbf{0.543\pm0.02}$ & $\mathbf{0.862\pm0.02}$ & $\mathbf{0.162\pm0.02}$ & ${0.224\pm0.02}$ & $\mathbf{0.309\pm0.02}$ & $\mathbf{0.423\pm0.02}$\\[0pt]
\end{tabular}
\caption{The methods' test set NMSE on the fastMRI multi-coil brain dataset with standard errors. The double lines separate the type of training data available and bold font is used to denote the best performance within each category. Methods that use BM3D could not be included because the NMSE was computed in k-space and BM3D acts on the magnitude image, so the complex phase is not retained. Table \ref{tab:ssim}  shows a similar table for SSIM. \label{tab:perf}}
\end{table*}


\subsection{Evaluation of Denoising VarNet}

To evaluate the performance of the proposed Denoising VarNet architecture, we trained the best-case baseline for Standard VarNet with 10 cascades and the Denoising VarNet with 5 cascades, so that they had roughly the same number of parameters. Fig. \ref{fig:arch} shows that Denoising VarNet outperforms Standard VarNet on the test set for all considered $R_\Omega$ and $\sigma_n$, especially for more challenging acceleration factors and noise levels.

\subsection{Robustness to $\alpha$ \label{sec:alpha_rob}}

To evaluate the robustness to $\alpha$, we trained Noisier2Full, Robust SSDU and their weighted variants for $\alpha \in \{0.05, 0.25, 0.5, 0.75, 1, 1.25, 1.5, 1.75 \} $. We focused solely on the case where $R_\Omega = 8$ and $\sigma_n = 0.06$. The performance on the test set is shown in Fig. \ref{fig:alpha_line}, which shows that the weighted versions are considerably more robust. The weighted and unweighted minima were at $\alpha = 1$ and $1.25$ for Noisier2Full and $\alpha = 0.75$ and $0.5$ for Robust SSDU respectively. We employed these values of $\alpha$ for all experiments in Section \ref{sec:task_A_res} and \ref{sec:task_B_res}; we assumed that the tuned $\alpha$ at $R_\Omega = 8$ and $\sigma_n = 0.06$ is a reasonable approximation of the optimum for every every evaluated $R_\Omega$ and $\sigma_n$. 

\subsection{Task A: Fully sampled, noisy training data \label{sec:task_A_res}}

Rows 3-5 of Table \ref{tab:perf} show the how the test set NMSE of networks trained on fully sampled, noisy data compares with the fully-supervised benchmark.  Supervised w/o denoising's performance significantly degrades as $\sigma_n$ increases: for $R_\Omega = 8$ and $\sigma_n = 0.08$, Supervised w/o denoising's test set loss is approximately double that of the fully-supervised benchmark. In contrast, Noisier2Full consistently performs similarly to the benchmark: its NMSE is within 0.008 for all $\sigma_n$ and $R_\Omega$. The performance of Unweighted Noisier2Full was slightly poorer than the weighted version, especially for high noise levels at the more challenging acceleration factor $R_\Omega = 8$.   Two reconstruction examples are shown in Fig. \ref{fig:examples_full}. Here, and throughout this paper, the example reconstructions show the image domain RSS cropped to a central $320 \times 320$ region. The k-space NMSE and SSIM are also shown. Appendix \ref{app:tab} shows the mean SSIM on the test set for all methods.

\begin{figure*}[t]
\centering
	\includegraphics[width=\textwidth]{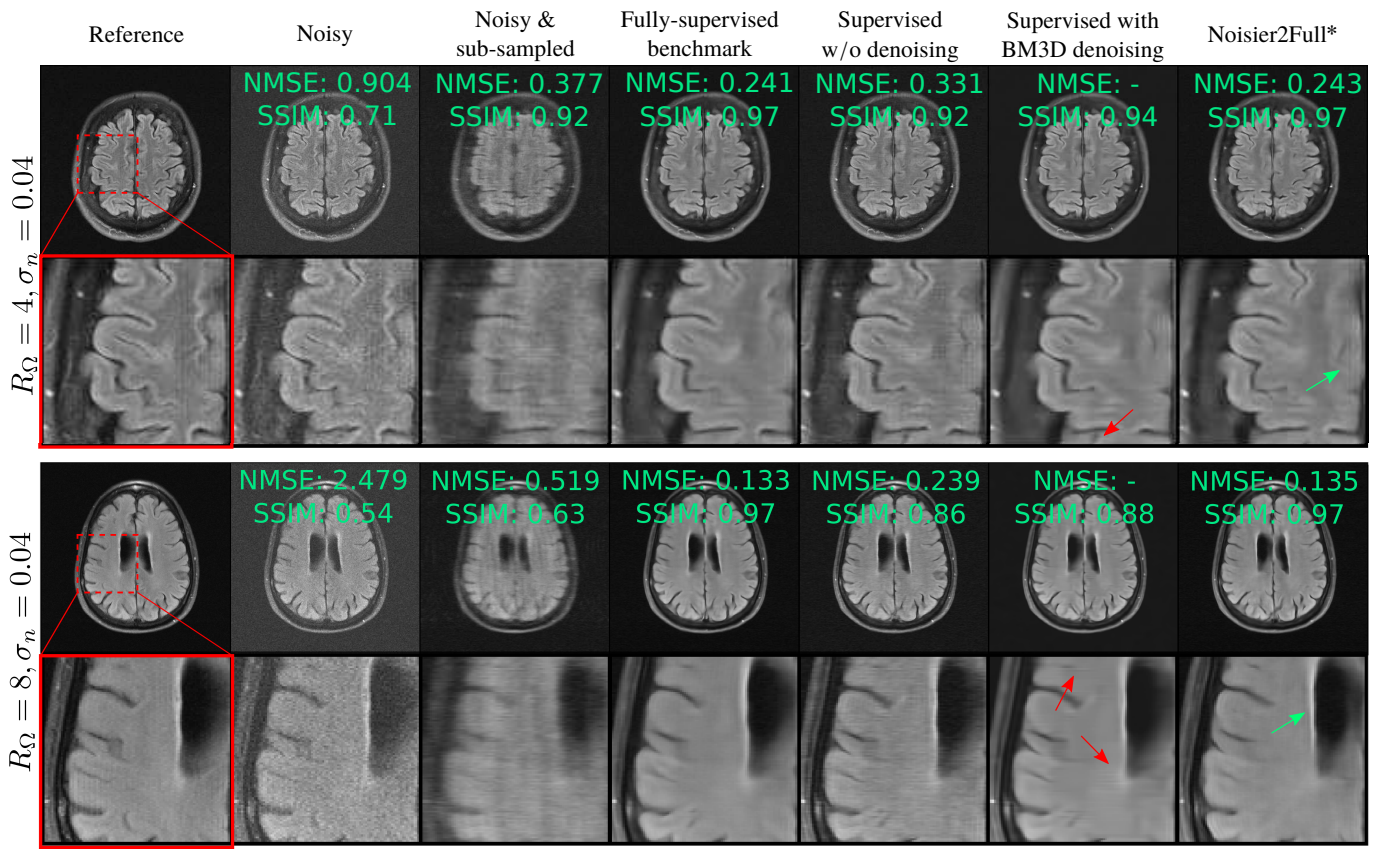}
	\caption{ Reconstructions when fully sampled, noisy  data is available for training. ``Noisy" and ``Noisy \& sub-sampled" refer to the RSS reconstruction of $y_{0,s} + n_s$ and $M_{\Omega_s} (y_{0,s} + n_s)$ respectively.  While there is clear noise in Supervised w/o denoising's reconstruction, the proposed method, which is indicated with an asterisk, perform very similarly to the fully-supervised benchmark. The red arrows show artifacts for Supervised with BM3D and the and green arrows show improved recovery and contrast of fine features for Noisier2Full \label{fig:examples_full} }
\end{figure*}


\begin{figure*}[t]
\centering
	\includegraphics[width=\textwidth]{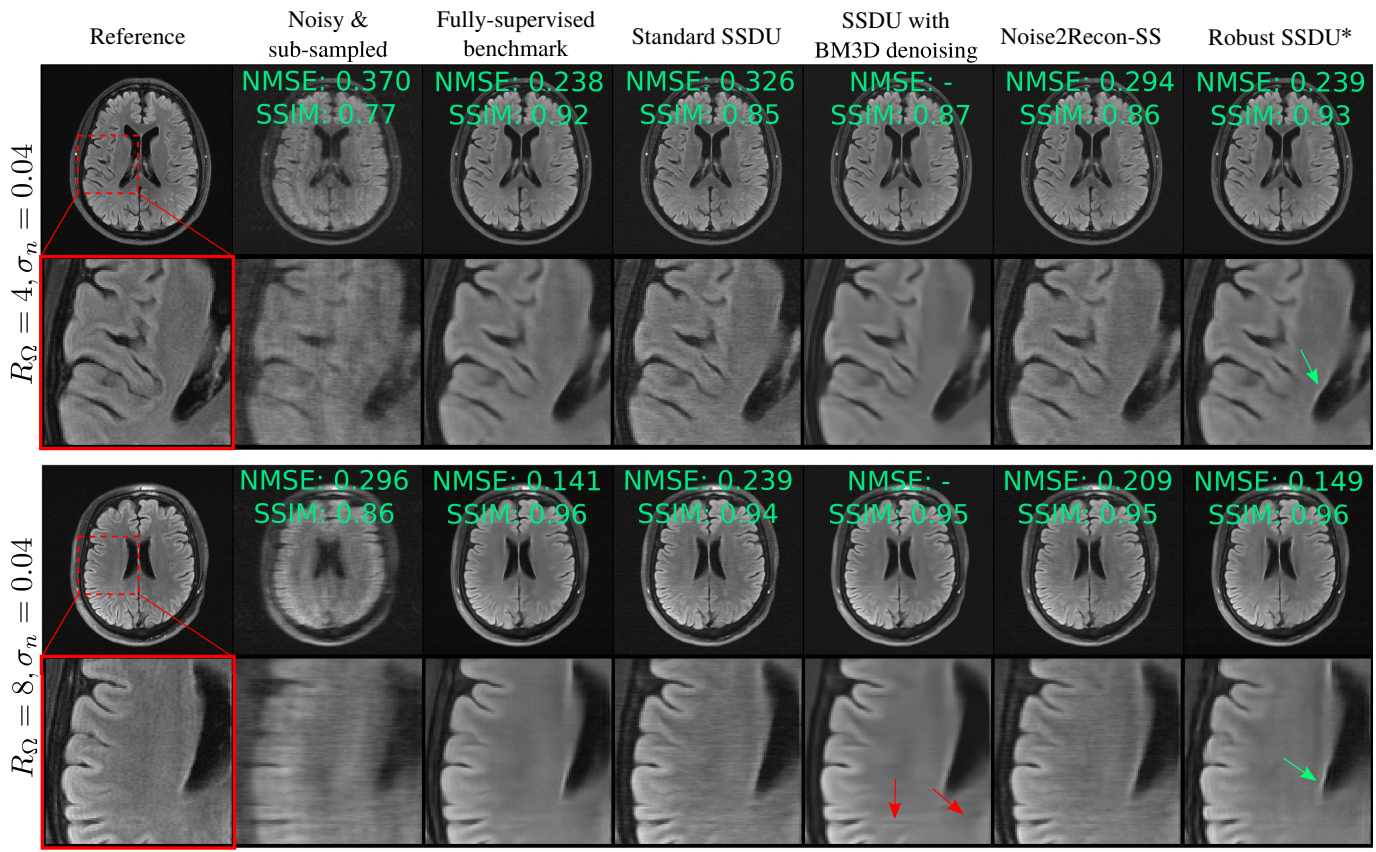}
	\caption{ Example reconstructions for networks trained on noisy, sub-sampled data. The proposed method Robust SSDU, highlighted with an asterisk, perform very similarly to the fully-supervised benchmark, even at $R_\Omega = 8$. \label{fig:examples_ss}}
\end{figure*}


\begin{figure}[]
\centering
	\includegraphics[width=0.5\textwidth]{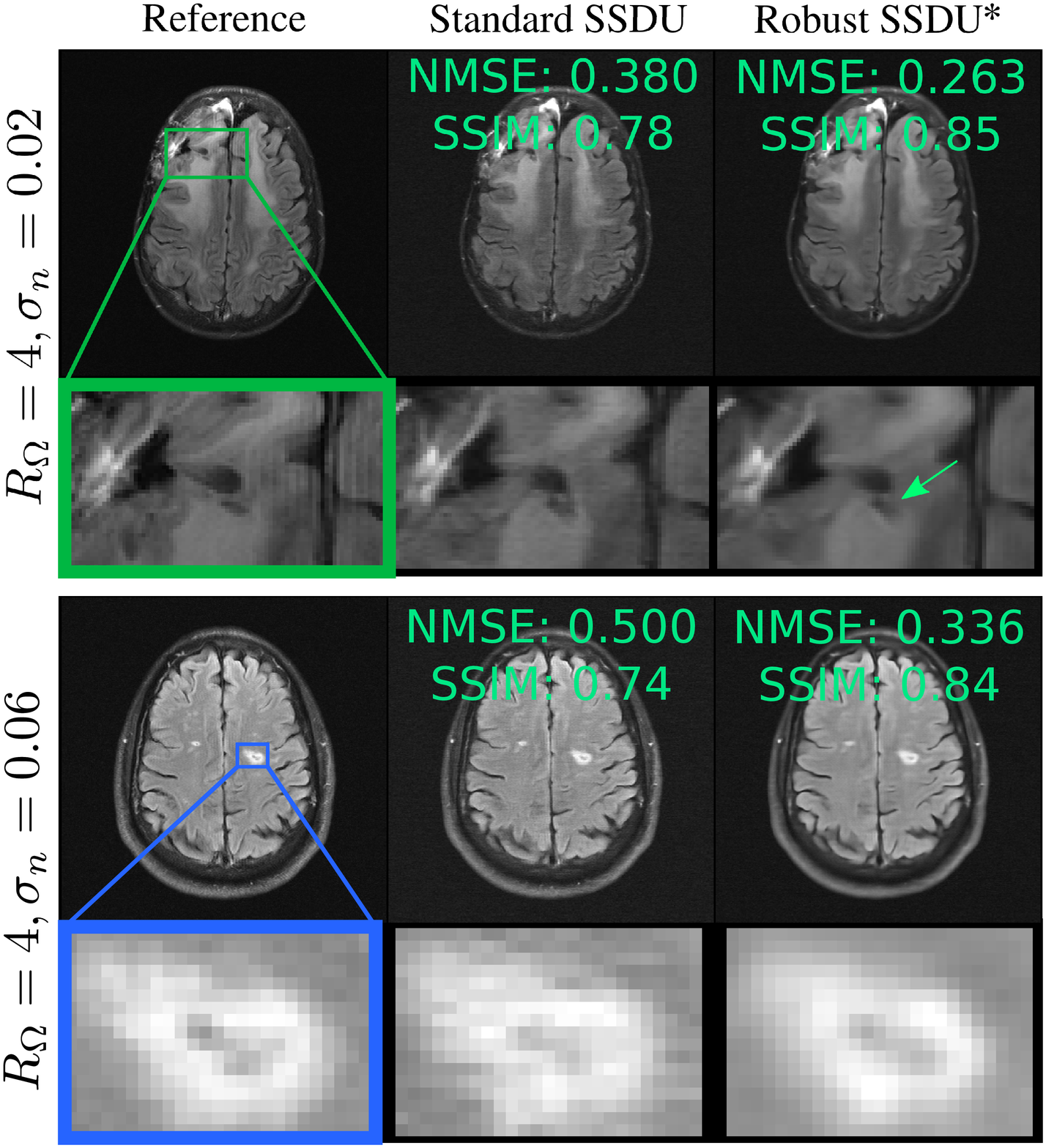}
	\caption{Clinical regions of interest annotated via fastMRI+ \cite{zhao2022fastmri+}. The top image shows  a resection cavity and the bottom shows a lacunar infarct. The proposed method Robust SSDU has improved sharpness compared to Standard SSDU, which has reconstruction errors arising from measurement noise. \label{fig:path}}
\end{figure}

\subsection{Task B: Sub-sampled, noisy training data \label{sec:task_B_res}}

Rows 6-9 of Table \ref{tab:perf} show the test set loss for the methods designed for sub-sampled, noisy training data. Robust SSDU performed within 0.012 of the fully-supervised benchmark, despite only having access to noisy, sub-sampled training data.  Noise2Recon-SS performs well in some cases, particularly at $R_\Omega = 4$, but is consistently outperformed by both variants of Robust SSDU.  Fig. \ref{fig:examples_ss} shows example reconstructions, demonstrating similar performance to the fully-supervised benchmark qualitatively. Fig. \ref{fig:path} compares Standard SSDU and Robust SSDU using clinical expert bounding boxes from fastMRI+ \cite{zhao2022fastmri+}, which shows that the proposed method has substantially enhanced pathology visualization.  For 2D Bernoulli sampling we found that a lower $R_\Lambda$ and $\alpha$ achieved better performance in practice: we used $R_\Lambda = 1.5$ and $\alpha = 0.5$. Fig. \ref{fig:bern} compares Standard SSDU and Robust SSDU for 2D Bernoulli sampled data at $R_\Omega = 4$ and $\sigma_n = 0.04$, showing that the denoising effect is not specific to column-wise sampling.



%
%

\section{Discussion and conclusions \label{sec:disc}}

Fig. \ref{fig:arch} shows that the proposed Denoising VarNet consistently outperforms the Standard VarNet architecture.  We understand this to be a consequence of the difference between the distributions of errors due to sub-sampling or measurement noise: the Standard VarNet removes both contributions to the error in a single U-net per cascade, while the Denoising VarNet simplifies the task by decomposing the contributions to the error,  so that each of the two U-nets per cascade are specialized for the two distinct error distributions. 

 The improvement in robustness for the weighted versions, shown in Fig. \ref{fig:alpha_line}, is especially prominent for small $\alpha$. For instance, at $\alpha = 0.05$, the unweighted variant of Noisier2Full is 0.50 from the benchmark, while the weighted variant is only 0.04 away. For large $\alpha$ the $\alpha$-based weighting is closer to 1, so weighted Noisier2Full tends to the unweighted method and the difference in performance is small. For instance, when $\alpha = 1.75$, the $\alpha$-based weighting is $1.33$, so has a relatively marginal effect. Although the performance of the methods are reasonably similar for tuned $\alpha$, we recommend using the weighted version in practice due to its improved robustness to $\alpha$. We emphasize that $\alpha$ tuning is only possible here because the noise and sub-sampling are simulated retrospectively; if the data were prospectively noisy and sub-sampled it would not possible to evaluate the fidelity of the estimate and the ground truth. Robustness to hyperparameters such as $\alpha$ is therefore of great importance for the method's usefulness in practice.   

The examples in figures \ref{fig:examples_full}, \ref{fig:examples_ss} and \ref{fig:path} show that proposed methods are qualitatively very similar to the fully-supervised benchmark, and substantially improve over methods without denoising, whose reconstructions are visibly corrupted with measurement noise. The examples exhibit some loss of detail and blurring at tissue boundaries, especially at $R_\Omega = 8$. However, the extent of detail loss is similar in the benchmark, indicating that the loss of detail is not a limitation of the proposed methods. Rather, the qualitative performance is limited by the other factors such as the architecture, dataset and choice of loss function. This can also be explained in part by noting that the high-frequency regions of k-space, which provide fine details, typically have smaller signal so are particularly challenging to recover in the presence of significant measurement noise. 

Table \ref{tab:perf} shows that the NMSE of the noisy, sub-sampled input to the network is \textit{lower} for the higher acceleration factor. This counter-intuitive incidental finding can be understood by noting that the spectral density is typically highly concentrated towards the center, so much of k-space has a small magnitude. Therefore, even for moderate noise levels, zero may be closer to the ground truth than the noisy data, so masking out such regions may improve the NMSE. This is also reflected in the NMSE scores of the reconstructed images. However, we note that this effect is not generally reflected in the qualitative performance on the methods, which we found more frequently exhibit oversmoothing and artifacts for higher acceleration. We believe this to be because the masked data is biased, so it is more difficult to achieve a high quality qualitative performance in practice.

The pseudo-denoising effect described in Section \ref{sec:theory_prop} is visible in Fig. \ref{fig:examples_full}, showing less noise in Supervised w/o denoising than Noisy.  
Table \ref{tab:perf} shows that Standard SSDU performs very similarly to Supervised w/o denoising quantitatively, and exhibits a similar pseudo-denoising effect in Fig. \ref{fig:examples_ss}. 
\begin{figure}[t]
\centering
	\includegraphics[width=0.48\textwidth]{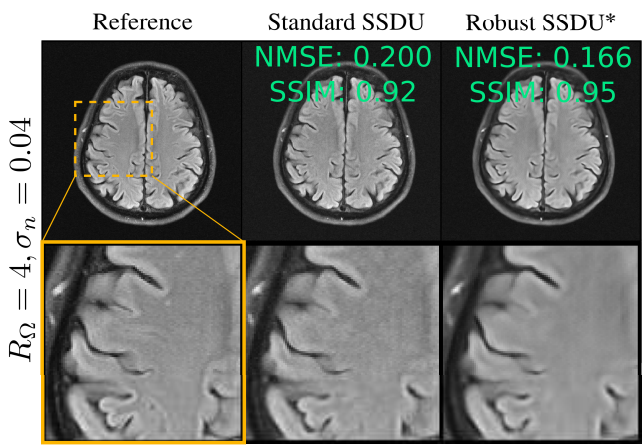}
	\caption{ Example reconstruction for 2D Bernoulli sampling. For Standard SSDU the test set NMSE and SSIM was 0.383 and 0.72 respectively, and for Robust SSDU the test set NMSE and SSIM was  0.316 and 0.75 respectively. \label{fig:bern}}
\end{figure}

\begin{figure*}[!t]
\centering
	\includegraphics[width=\textwidth]{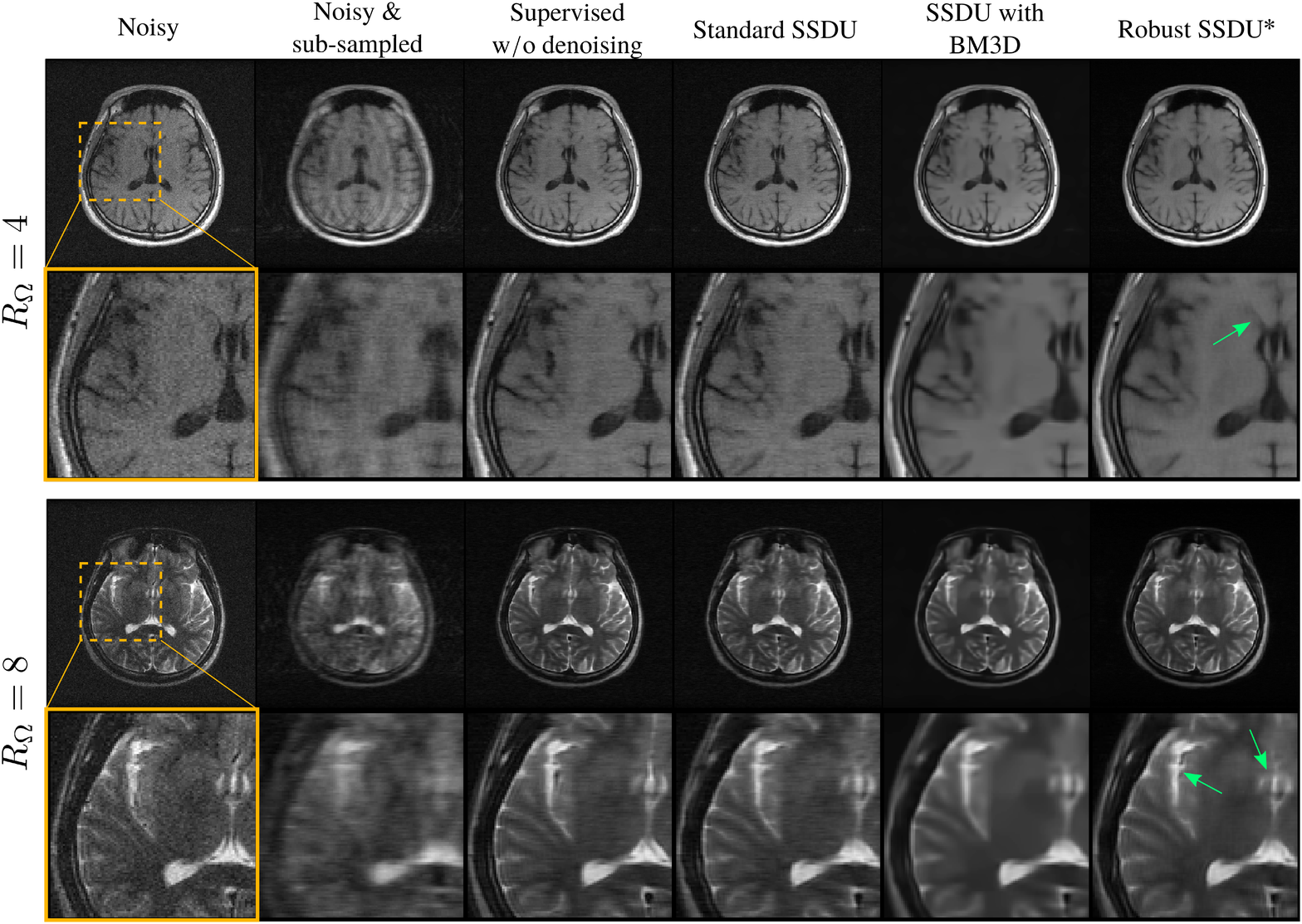}
	\caption{ Qualitative performance of the proposed method on the prospectively noisy
low-field dataset M4Raw. While SSDU with BM3D and Robust SSDU both demonstrate a
denoising effect, Robust SSDU exhibits improved contrast and visibly sharper boundaries,
highlighted by the colored arrows.  \label{fig:examples_ss_m4raw}}
\end{figure*}
 
Although Noise2Recon-SS improves over Standard SSDU, there is a substantial difference between its performance and the proposed Robust SSDU both qualitatively and quantitatively. In \cite{desai2023noise2recon},  Noise2Recon-SS was not compared with a fully-supervised benchmark; it was only shown to have improved performance compared to Standard SSDU, consistent with the results  here. The experimental evaluation in  \cite{desai2023noise2recon} focused on robustness to Out of Distribution (OOD) shifts, where the training and inference measurement noise variances not necessarily the same. Another difference is that Noise2Recon-SS's simulated noise in \cite{desai2023noise2recon} had standard deviation randomly selected from a fixed range, while the experiments here fixed the simulated noise standard deviation so that it could be properly compared with the proposed methods.  

Robust SSDU requires only a few additional cheap computational steps compared to standard training: the addition or multiplication of the further noise and sub-sampling mask respectively, and the $\alpha$-based correction at inference. Accordingly, the compute time and memory requirements of the proposed methods was found to be very similar to Supervised w/o denoising or Standard SSDU.  In contrast, Noise2Recon-SS uses both $M_{\Lambda_t} y_t$ and $y_t + M_{\Omega_t} \nt_t$ as the network inputs at training, so requires twice as many forward passes to train the network compared to Robust SSDU.  Accordingly, we found that Noise2Recon-SS required approximately twice as much memory and took around two times longer per epoch as the proposed methods. 

In general, Supervised with BM3D and SSDU with BM3D both performed well qualitatively. We also found that in many cases these methods had an mean SSIM that exceeded even the fully-supervised benchmark: see Appendix \ref{app:tab} for a detailed discussion. However, for some images, such as those shown with the red arrows of figures 5 and 6, these methods generated potentially clinically misleading artifacts. We believe this to be a consequence of the mismatch between its Gaussian noise model and the actual error of the RSS estimate, which can lead to unreliable noise removal, especially at high $R_\Omega$. We also found that SSDU with BM3D often led to more oversmoothing and less crisp tissue boundaries than Robust SSDU, which is particularly prominent in the M4Raw examples of Fig. \ref{fig:examples_ss_m4raw}. Another disadvantage is the computational expense of the BM3D algorithm:  we found that the reconstruction time of SSDU with BM3D was around 100 times longer per slice than Robust SSDU at inference.

Another existing method designed for noisy, sub-sampled training data is Robust Equivariant Imaging (REI) method \cite{chen2022robust, chen2021equivariant}.  We did not compare with REI as it was  designed for reconstruction tasks with a fixed sampling pattern: the $\Omega_t$ was the same for all $t$. This sampling set assumption  is central to their  use of equivariance, and contrasts with the methods proposed here, which assume that the sampling mask is an instance of a random variable that satisfies $p_j > 0$ everywhere. However, REI's suggestion to use Stein's Unbiased Risk Estimate   (SURE) \cite{Stein1981}  to remove measurement noise would be feasible in combination with SSDU and warrants further investigation in future work. 

The theoretical work presented in this paper only applies to the case of $\ell_2$ minimization, which can lead to blurry reconstructions. However, it has been established that Standard SSDU can be applied with other losses such as an entry-wise  mixed $\ell_1$-$\ell_2$ loss in k-space \cite{yaman2020self}. We have found that Robust SSDU with an $\ell_2$ loss on $\Lambda \cap \Omega$ and mixed $\ell_1$-$\ell_2$ loss on $\Omega \setminus \Lambda$ also performs competitively with a suitable benchmark in practice (results not shown for brevity). Future work includes establishing whether Robust SSDU can be modified to be applicable to other loss functions, including potentially losses on the RSS image. 


The methods presented here also assume that the distribution of $M_\Omega$ is fixed; a modification of the method for dealing with a range of sub-sampling patterns and acceleration factors is a potential avenue for future work. It would also be desirable to develop an approach that automatically tunes $\alpha$ and the distribution of $\Mt$, whose optimal values are specific to the noise model, $M_\Omega$ distribution and dataset. 

Additive Noisier2Noise was designed for Gaussian noise, and is not expected to perform well for other measurement noise distributions. Future work includes extending the framework to other  distributions  and sources of error such as more complex system noise or physiological motion, which has a more complex distribution that may itself be learned \cite{10404226,spieker2023deep}

\section{Appendices}

\subsection{Proof of SSDU variant on $\Mo$} \label{app:mo}

 
This appendix proves Claim \ref{clm:robust_ssdu} using a similar approach to Appendix B of \cite{millard2022theoretical}. Minimization according to \eqref{eqn:LSSDU} yields a network that satisfies
\begin{equation}
\Eds[ \Mo (f_{\theta^*} (\Yt) - Y) | \Yt ] = 0 \label{eqn:meq0}
\end{equation}
We split the conditional expectation into two cases: $ \Yt_j \neq 0$ and $ \Yt_j = 0$. Throughout this paper, $m_j$ and $\mt_j$ refer to the $j$th diagonal of $\Mo$ and $\Mt$ respectively.

\textit{Case 1} ($\Eds[ m_j (f_{\theta^*} (\Yt)_j - Y_j) | \Yt_j  \neq 0]$): When $\Yt_j \neq 0$, the measurement model implies that $m_j=1$  and $Y_j = Y_{0,j} + N_j$. Therefore
\begin{multline}
\Eds[ m_j (f_{\theta^*} (\Yt)_j - Y_j) | \Yt_j  \neq 0] \\ = \Eds[ f_{\theta^*} (\Yt)_j - Y_{0,j} - N_j | \Yt_j  \neq 0] \label{eqn:case1}
\end{multline}

\textit{Case 2} ($\Eds[ m_j (f_{\theta^*} (\Yt)_j - Y_j) | \Yt_j  = 0]$):
We can use the result derived from equation (27) to (29)  in \cite{millard2022theoretical}, with $Y_{0,j}$ replaced by $Y_{0,j} + N_j$:
\begin{multline}
\Eds[ m_j (f_{\theta^*} (\Yt)_j - Y_j) | \Yt_j  = 0] \\ = \Eds[ f_{\theta^*} (\Yt)_j - Y_{0,j} - N_j | \Yt_j  = 0] \cdot (1-k_j) \label{eqn:case2}
\end{multline}
where 
\begin{equation}
 k_j = \Pds [Y_j = 0 |\Yt_j = 0 ] = \frac{1 - p_j}{1 - \pt_j p_j}. \label{eqn:k_def}
\end{equation}

\textit{Combining Cases 1 and 2}: Consider the candidate
\begin{align}
 & \Eds[ m_j (f_{\theta^*} (\Yt)_j - Y_j) | \Yt_j ] \nonumber  
\\ &= \{1-k_j(1 - \mt_j m_j)\} \Eds[ f_{\theta^*} (\Yt)_j - Y_{0,j} - N_j| \Yt_j]. \label{eqn:Em}
\end{align}
To verify that this expression is correct, we can check that it is consistent with Cases 1 and 2. For Case 1: if $\Yt_j \neq 0$, $\mt_j m_j = 1$ and the term in curly brackets is $1$, so \eqref{eqn:Em} is consistent with \eqref{eqn:case1}. For Case 2: if $\Yt_j = 0$, $\mt_j m_j = 0$ and  the term in curly brackets is $1 - k_j$, so \eqref{eqn:Em} is consistent with \eqref{eqn:case2} as required. By \eqref{eqn:meq0},
$$ \{1-k_j(1 - \mt_j m_j)\} \Eds[ f_{\theta^*} (\Yt)_j - Y_{0,j} - N_j| \Yt_j]  = 0 $$
The term in the curly brackets is non-zero for all $j$ if $1-k_j$ is non-zero for $j \notin \Omega \cap \Lambda$, which true when \eqref{eqn:mask_req1} and \eqref{eqn:mask_req2} hold, where we note that the special case $\pt_j = p_j = 1$ is also allowed since $\mt_j m_j = 1$ always. 
Given this assumption, dividing by the term in the curly brackets:
\begin{equation}
\Eds[ f_{\theta^*} (\Yt)_j - Y_{0,j} - N_j| \Yt_j] = 0.
\end{equation}
Vectorizing gives the required result. \hfill $\blacksquare$

\subsection{Proof of weighted Noisier2Full}
\label{app:Nr2N}
To compute the unknown
$$
\dt \Eds \left[ \left\|\hat{Y}_{Nr2F} - Y_0 \right \|^2_2  | \Yt \right]
$$
in terms of the known $Y_0 + N$, we compute the contributions to the loss in $\Omega$ and $\Omega^c$ separately, shown in lemmas \ref{lem:Nr2N} and \ref{lem:Nr2N_2} respectively.

\begin{lemma} \label{lem:Nr2N} Consider the random variables $Y = \Mo(Y_0 + N)$ and $\Yt = Y + \Mo \Nt$, where $N$ and $\Nt$ are zero-mean Gaussian distributed with variances $\sigma_n^2$ and $\alpha^2 \sigma_n^2$ respectively. For an arbitrary function $f_{\theta}$,
\begin{multline}
	 \dt \Eds \left[ \left\| \Mo \left(\Yhn - Y_0\right)  \right \|^2_2  | \Yt \right] \\ = \dt \Eds \left[ \left\|\A \Mo (f_\theta ( \Yt) - Y) \right \|^2_2 | \Yt  \right ]. \label{eqn:add_n2n_clm}
\end{multline} 
\end{lemma}
\begin{proof}
Using $\Mo \Yt =\Mo ( Y +  \Nt )$ and $\Mo Y_0 =\Mo ( Y -  N )$, the left-hand-side of \eqref{eqn:add_n2n_clm} is
\begin{align}
& \dt \Eds \left[ \left \|\Mo \left( \frac{(1+\alpha^2)f_\theta ( \Yt) - \Yt}{\alpha^2} - Y_0 \right) \right \|^2_2 | \Yt \right] \nonumber \\ 
&=\dt \Eds \left[ \left \|\Mo \left( \frac{(1+\alpha^2)f_\theta ( \Yt) - Y - \Nt}{\alpha^2} - Y + N \right) \right \|^2_2 | \Yt \right]  \nonumber  \\ 
&=\dt \Eds \left[ \left \|\Mo \left( \A (\ft - Y) + N - \frac{\Nt}{\alpha^2} \right )  \right \|^2_2 | \Yt \right]   \nonumber  \\
&= \dt \Eds \left[ \left\|\A \Mo (f_\theta ( \Yt) - Y) \right \|^2_2 | \Yt  \right ]  \nonumber \\ & \hspace{0.8cm} + \A \dt \Eds \left[ 2 \ft^H \Mo \left  ( N - \frac{\Nt}{\alpha^2} \right)  | \Yt  \right ] \label{eqn:n2n_proof_mid}
\end{align}
where all the terms in the expansion of the $\ell_2$ norm in the last step that are not dependent on $\theta$ have been zeroed by $\dt$. Now we show that the second term on the right-hand-side of \eqref{eqn:n2n_proof_mid} is zero. Lemma 3.1 from \cite{moran2020noisier2noise}  shows that
\begin{equation}
\Eds [ \Mo \Nt | \Yt ] = \alpha^2 \Eds [\Mo N | \Yt ], \label{eqn:nr2n_cond_exp}
\end{equation}
where $\Mo$ is included as the result only applies to sampled terms. We note that the right hand side of \eqref{eqn:nr2n_cond_exp} scales according to the \textit{variance} of the noise rather than the perhaps more intuitive standard deviation. Following \cite{moran2020noisier2noise}, \eqref{eqn:nr2n_cond_exp} can be proven by computing the probability $\Pds [N_j = n | \Yt_j, j \in \Omega]$: 
\begin{multline*}
\Pds [N_j = n | \Yt_j, j \in \Omega] = \Pds[N_j=n] \Pds[\Nt_j=\Yt_j - Y_{0,j} - n_j] \\
\propto \exp \left(- \frac{n^2}{2\sigma^2} \right) \exp \left(  - \frac{(\Yt_j - Y_{0,j} - n)^2}{2\alpha^2 \sigma^2} \right).
\end{multline*}
After some algebraic manipulation not shown here for brevity, this distribution can be shown to have mean $(\Yt_j - Y_{0,j})/(1 + \alpha^2)$. A similar computation for $\Nt_j $ yields a mean of $\alpha^2 (\Yt_j - Y_{0,j})/(1 + \alpha^2)$, giving the $j$th entry of the relationship stated in \eqref{eqn:nr2n_cond_exp} conditional on $j \in \Omega$. Since for the alternative $j \notin \Omega$ both sides are trivially zero,  \eqref{eqn:nr2n_cond_exp} is correct for all indices.

Applying \eqref{eqn:nr2n_cond_exp} to the right hand side of \eqref{eqn:n2n_proof_mid} gives
\begin{multline}
 \Eds \left[ \ft^H \Mo \left( N - \frac{\Nt}{\alpha^2} \right)  | \Yt  \right ] \\  =\ft^H \Eds \left[ \Mo \left( N - \frac{\Nt}{\alpha^2} \right)  | \Yt  \right ] = 0
\end{multline}
where the conditional dependence on $\Yt$ allows the removal of $\ft$ from the expectation. Therefore the right-hand-side of  \eqref{eqn:n2n_proof_mid} equals the right-hand-side of   \eqref{eqn:add_n2n_clm} as required.
\end{proof}

\begin{lemma} \label{lem:Nr2N_2} Consider the random variables $Y$ and $\Yt$ as defined in Lemma \ref{lem:Nr2N}. For an arbitrary function $f_{\theta}$,
\begin{multline}
	\dt \Eds \left[ \left\| \Moc (\Yhn - Y_0 ) | \Yt \right \|^2_2   \right] \\ = \dt \Eds \left[ \left\|\Moc (f_\theta ( \Yt) - Y_0 - N) \right \|^2_2 | \Yt  \right ]. \label{eqn:add_n2n_clm_moc}
\end{multline}
\end{lemma}
\begin{proof}
Using $\Moc \Mo = 0$, $\Moc \Moc = \Moc$ and the definition of $\Yhn$ in  \eqref{eqn:n2f_yhat}, we have $\Moc \Yhn = \Moc \ft$. Therefore the left-hand-side of \eqref{eqn:add_n2n_clm_moc} is 
\begin{align}
&\dt \Eds \left[ \left\|\Moc (f_\theta ( \Yt) - Y_0) \right \|^2_2 | \Yt  \right ] \nonumber \\
&= \dt \Eds \left[ \left\|\Moc (f_\theta ( \Yt) - Y_0 - N + N) \right \|^2_2 | \Yt  \right ] \nonumber \\
&= \dt \Eds \left[ \left\|\Moc (f_\theta ( \Yt) - Y_0 - N)\right \|^2_2 + 2 \ft ^H \Moc N  | \Yt  \right ]
 \label{eqn:a}
\end{align}
where again all the terms not dependent on $\theta$ have been zeroed by $\dt$. The second term is 
\begin{equation}
\Eds \left[\ft ^H \Moc N  | \Yt  \right ] = \ft ^H \Eds \left[\Moc N  | \Yt  \right ] = 0 
\end{equation}
where, as in \eqref{eqn:motiv_notin}, we have used the independence of $N$ from $\Yt$ when $j \notin \Omega$. Therefore \eqref{eqn:a} equals the right-hand-side of \eqref{eqn:add_n2n_clm_moc} as required.
\end{proof}

To find the $\ell_2$ error of $\Yhn$, we use $\Mo + \Moc = \eye$ and sum the results from lemmas \ref{lem:Nr2N} and \ref{lem:Nr2N_2}:
\begin{align*}
	 &\dt \Eds \left[ \left\|\hat{Y}_{Nr2F} - Y_0 \right \|^2_2  | \Yt \right] \\ &
	 = \dt \Eds \left[ \left\|(\Mo + \Moc)(\hat{Y}_{Nr2F} - Y_0 ) \right \|^2_2  | \Yt \right] \\ &	
	 =\dt  \Eds \left[ \left\|\left(\A \Mo + \Moc \right)(\ft - Y_0 - N ) \right \|^2_2  | \Yt \right]
\end{align*}
as required. 
\subsection{Proof of Robust SSDU weighting}
\label{app:rssdu}

Analogous to Appendix \ref{app:Nr2N}, to compute the unknown
$$
\dt \Eds \left[ \left\|\Yhr - Y_0 \right \|^2_2  | \Yt \right]
$$
in terms of the known sub-sampled, noisy $Y$, we compute the contributions to the loss from $\Lambda \cap \Omega $ and $(\Lambda \cap \Omega)^c$ separately. For the contribution from $\Lambda \cap \Omega $, an identical approach to the proof in Lemma \ref{lem:Nr2N} can be used with $\Omega$ replaced by $\Lambda \cap \Omega $, so that 
\begin{multline}
	 \dt \Eds \left[ \left\| \Mlo \left(\Yhr - Y_0\right)  \right \|^2_2  | \Yt \right] \\ = \dt \Eds \left[ \left\|\A \Mlo (f_\theta ( \Yt) - Y) \right \|^2_2 | \Yt  \right ].  \label{eqn:mloc_r}
\end{multline}
The following lemma shows how the remaining loss, which is computed on $\Omega \setminus \Lambda$, can be used to estimate the target ground truth loss, which is over $(\Lambda \cap \Omega)^c$.

\begin{lemma} \label{lem:}
Consider the random variables $Y = \Mo (Y_0 + N)$ and $\Yt = \Mlo (Y +  \Nt)$, where $N$ and $\Nt$ are zero-mean Gaussian distributed with variances $\sigma_n^2$ and $\alpha^2 \sigma_n^2$ respectively. For an arbitrary function $f_{\theta}$,
\begin{multline}
\dt \Eds\left[  \left\|\Mloc (\Yhr - Y_0) \right \|^2_2 | \Yt  \right ] \\ = 
 \dt \Eds\left[  \left \|  \mathcal{P}^{1/2} \Moml (f_\theta ( \Yt) - Y ) \right \|^2_2 | \Yt \right] , \label{eqn:lem3_res}
\end{multline}
where $\mathcal{P}$ is defined in \eqref{eqn:p_def}.
\end{lemma}
\begin{proof}
Since $\Mloc \Yhr = \Mloc \ft$, the left-hand-side of \eqref{eqn:lem3_res} is 
\begin{multline}
\dt \Eds\left[  \left\|\Mloc (\ft   - Y_0) \right \|^2_2 | \Yt  \right ] \\ = 
\dt \Eds \left[ \left\|\Mloc (f_\theta ( \Yt) - Y_0 - N) \right \|^2_2 | \Yt  \right ], 
\end{multline}
where Lemma \ref{lem:Nr2N_2} with $\Omega^c$ replaced by $(\Lambda \cap \Omega)^c$ has been used. Using $|\cdot|^2$ to denote the entry-wise magnitude squared, we can write 
\begin{multline}
 \dt \Eds \left[ \left\|\Mloc (f_\theta ( \Yt) - Y_0 - N) \right \|^2_2 | \Yt  \right ]  \\ 
 =  \dt \Eds \left[1_q^T  \Mloc |f_\theta ( \Yt) - Y_{0} - N|^2   | \Yt  \right ], \label{eqn:1q}  
\end{multline} 
where $ 1_q$ is a $q$-dimensional vector of ones. Eqn. (32) from \cite{millard2022theoretical} shows that  the conditional expectation of $\ft - Y$ on $\Mloc$ and $\Moml$ is related by a factor $\Pc$. By repeating that derivation with all instances of $\ft - Y$ trivially replaced with $|\ft - Y_0 - N|^2$, a similar relationship can be derived for the latter, yielding the same $\Pc$ factor: see \cite{millard2022theoretical}. In brief, if the $j$th entry $\Yt_j$ is not zero then the $j$th diagonal of $\Moml$ is zero, so
\begin{align}
	\Eds\left[   | \Moml (f_\theta ( \Yt) - Y_0 - N )  |_j^2 | \Yt_j \neq 0 \right] = 0. \label{eqn:foo}
\end{align}
When the $j$th entry of $\Yt_j$ \textit{is} zero, 
\begin{multline}
	\Eds\left[   | \Moml (f_\theta ( \Yt) - Y_0 - N )  |_j^2 | \Yt_j = 0 \right] \\ = \Eds\left[ \Pc_{jj}^{-1}  | f_\theta ( \Yt) - Y   |_j^2 | \Yt_j = 0 \right]. \label{eqn:bar}
\end{multline}
See (31) of \cite{millard2022theoretical} for a detailed derivation. Combining both cases from  \eqref{eqn:foo} and \eqref{eqn:bar}, 
\begin{multline}
	\Eds\left[   | \Moml (f_\theta ( \Yt) - Y_0 - N )  |_j^2 | \Yt_j\right ] \\ =  \Eds \left[ \Pc_{jj}^{-1}   | \Mloc  (f_\theta ( \Yt) - Y_{0} - N)|^2_j   | \Yt_j  \right ]. 
\end{multline}
Applying this result to \eqref{eqn:1q} by multiplying with $1_q^T$, and bringing the masks outside of the  entry-wise magnitude:
\begin{align*} 
&\dt \Eds \left[ 1_q^T \Mloc |f_\theta ( \Yt) - Y_{0} - N|^2   | \Yt  \right ]  \\
&=  \dt \Eds \left[ 1_q^T \Pc \Moml    |f_\theta ( \Yt) - Y_{0} - N|^2   | \Yt  \right ] \\
&= \dt \Eds\left[  \left \|  \mathcal{P}^{1/2} \Moml (f_\theta ( \Yt) - Y ) \right \|^2_2 | \Yt \right],
\end{align*}
as required.
\end{proof}

To find the $\ell_2$ error of $\Yhr$, we use $\Mlo + \Mloc = \eye$ and sum \eqref{eqn:mloc_r} and \eqref{eqn:lem3_res}:
\begin{align*}
	 &\dt \Eds \left[ \left\|\Yhr - Y_0 \right \|^2_2  | \Yt \right] \\ &
	 = \dt \Eds \left[ \left\|(\Mlo + \Mloc)(\hat{Y}_{Nr2F} - Y_0 ) \right \|^2_2  | \Yt \right] \\ &	
	 = \dt \Eds \left[ \left\|\left(\A \Mlo + \mathcal{P}^{1/2} \Moml \right)(\ft - Y) \right \|^2_2  | \Yt \right] 
\end{align*}
as required. 

\subsection{Table of SSIM on test set \label{app:tab}}

The mean SSIM on the magnitude images are shown in Table \ref{tab:ssim}. The SSIM of the proposed methods is comparable to the fully-supervised benchmark. However, in many cases, the methods that use BM3D outperform even the fully-supervised benchmark, implying that BM3D achieves a better SSIM than the machine learning based approach to denoising used in this paper. We emphasize that the entirely data-driven approaches were not trained to minimize for SSIM, and the SSIM would be expected to substantially improve if it was included in the loss function \cite{zbontar2018fastmri}. 

The methods that use BM3D have a considerably higher standard error, which indicates an substantially higher variation in the quality of the output. We believe that this is a consequence of the mismatch between BM3D's Gaussian noise model and the actual error of the RSS estimate, which leads to higher risk of oversmoothing and artifacts such as those shown in figures \ref{fig:examples_full} and \ref{fig:examples_ss}. 

\begin{table*}[th]
\centering
\fontsize{5.6pt}{8pt}\selectfont
\begin{tabular}{c||c|c|c|c||c|c|c|c}
\multicolumn{1}{c}{} & \multicolumn{4}{c}{\textsc{Acceleration factor} $R_\Omega=4$ } & \multicolumn{4}{c}{\textsc{Acceleration factor} $R_\Omega=8$ } \\ 
      & $\sigma_n=0.02$ & $\sigma_n = 0.04$ & $\sigma_n = 0.06$ & $\sigma_n = 0.08$ & $\sigma_n=0.02$ & $\sigma_n = 0.04$ & $\sigma_n = 0.06$ & $\sigma_n = 0.08$  \\\hhline{=#=|=|=|=#=|=|=|=}
Noisy \& sub-sampled  &  $0.76\pm0.006$ & $0.64\pm0.006$ & $0.50\pm0.006$ & $0.40\pm0.005$ & $0.72\pm0.008$ & $0.67\pm0.007$ & $0.57\pm0.006$ & $0.48\pm0.006$   \\[0pt] \hline
Fully-supervised benchmark & $\mathbf{0.83\pm0.007}$ & $\mathbf{0.75\pm0.006}$ & $\mathbf{0.63\pm0.006}$ & $\mathbf{0.52\pm0.005}$ & $\mathbf{0.75\pm0.008}$ & $\mathbf{0.77\pm0.007}$ & $\mathbf{0.73\pm0.007}$ & $\mathbf{0.67\pm0.006}$\\[0pt] \hhline{=#=|=|=|=#=|=|=|=}
Supervised w/o denoising & ${0.83\pm0.006}$ & ${0.70\pm0.006}$ & ${0.55\pm0.005}$ & ${0.43\pm0.005}$ & ${0.80\pm0.008}$ & ${0.74\pm0.006}$ & ${0.63\pm0.006}$ & ${0.52\pm0.005}$\\[0pt] \hline
Supervised with BM3D& $\mathbf{0.86\pm0.025}$ & $\mathbf{0.75\pm0.042}$ & $\mathbf{0.64\pm0.043}$ & $\mathbf{0.55\pm0.041}$ & $\mathbf{0.85\pm0.014}$ & $\mathbf{0.78\pm0.033}$ & ${0.69\pm0.039}$ & ${0.61\pm0.039}$\\[0pt] \hline
Unweighted Noisier2Full* & ${0.83\pm0.007}$ & $\mathbf{0.75\pm0.006}$ & ${0.63\pm0.006}$ & ${0.52\pm0.005}$ & $0.76\pm0.008$ & ${0.77\pm0.007}$ & $\mathbf{0.73\pm0.007}$ & $\mathbf{0.66\pm0.006}$\\[0pt] \hline
Noisier2Full* & ${0.82\pm0.007}$ & ${0.74\pm0.006}$ & ${0.62\pm0.006}$ & ${0.50\pm0.005}$ & ${0.74\pm0.008}$ & ${0.76\pm0.007}$ & ${0.72\pm0.007}$ & ${0.65\pm0.006}$\\[0pt] \hhline{=#=|=|=|=#=|=|=|=}
Standard SSDU & ${0.83\pm0.004}$ & ${0.69\pm0.004}$ & ${0.55\pm0.003}$ & ${0.43\pm0.003}$ & ${0.79\pm0.005}$ & ${0.74\pm0.004}$ & ${0.63\pm0.004}$ & ${0.52\pm0.003}$\\[0pt] \hline
SSDU with BM3D & $\mathbf{0.86\pm0.025}$ & $\mathbf{0.75\pm0.042}$ & $\mathbf{0.64\pm0.043}$ & $\mathbf{0.56\pm0.041}$ & $\mathbf{0.84\pm0.014}$ & $\mathbf{0.78\pm0.033}$ & ${0.69\pm0.039}$ & ${0.61\pm0.039}$\\[0pt] \hline
Noise2Recon-SS & ${0.83\pm0.006}$ & ${0.71\pm0.006}$ & ${0.56\pm0.005}$ & ${0.47\pm0.005}$ & ${0.79\pm0.008}$ & ${0.73\pm0.006}$ & ${0.66\pm0.006}$ & ${0.56\pm0.005}$\\[0pt] \hline
Unweighted Robust SSDU* & ${0.83\pm0.007}$ & $\mathbf{0.75\pm0.006}$ & ${0.62\pm0.006}$ & ${0.51\pm0.005}$ & ${0.75\pm0.008}$ & ${0.77\pm0.007}$ & $\mathbf{0.72\pm0.006}$ & $\mathbf{0.65\pm0.006}$\\[0pt] \hline
Robust SSDU* & ${0.83\pm0.007}$ & ${0.74\pm0.006}$ & ${0.62\pm0.006}$ & ${0.50\pm0.005}$ & ${0.75\pm0.008}$ & ${0.76\pm0.007}$ & $\mathbf{0.72\pm0.007}$ & $\mathbf{0.65\pm0.006}$\\[0pt] 
\end{tabular}
\caption{The methods' mean test set SSIM on the magnitude images with standard errors.   \label{tab:ssim}}
\end{table*}

\bibliographystyle{ieeetr}
\bibliography{library_manu}

\begin{thebibliography}{10}

\bibitem{bustin2020compressed}
A.~Bustin, N.~Fuin, R.~M. Botnar, and C.~Prieto, ``{From compressed-sensing to
  artificial intelligence-based cardiac MRI reconstruction},'' {\em Frontiers
  in cardiovascular medicine}, vol.~7, p.~17, 2020.

\bibitem{Pruessmann1999}
K.~P. Pruessmann, M.~Weiger, M.~B. Scheidegger, and P.~Boesiger, ``{SENSE:
  sensitivity encoding for fast MRI.},'' {\em Magnetic resonance in medicine},
  vol.~42, pp.~952--62, nov 1999.

\bibitem{Lustig2007}
M.~Lustig, D.~Donoho, and J.~M. Pauly, ``{Sparse MRI: The application of
  compressed sensing for rapid MR imaging},'' {\em Magnetic Resonance in
  Medicine}, vol.~58, pp.~1182--1195, dec 2007.

\bibitem{Ye2019}
J.~C. Ye, ``{Compressed sensing MRI: a review from signal processing
  perspective},'' {\em BMC Biomedical Engineering}, vol.~1, p.~8, dec 2019.

\bibitem{wang2016accelerating}
S.~Wang, Z.~Su, L.~Ying, X.~Peng, S.~Zhu, F.~Liang, D.~Feng, and D.~Liang,
  ``Accelerating magnetic resonance imaging via deep learning,'' in {\em 2016
  IEEE 13th International Symposium on Biomedical Imaging (ISBI)},
  pp.~514--517, 2016.

\bibitem{kwon2017parallel}
K.~Kwon, D.~Kim, and H.~Park, ``A parallel {MR} imaging method using multilayer
  perceptron,'' {\em Medical physics}, vol.~44, no.~12, pp.~6209--6224, 2017.

\bibitem{hammernik2018learning}
K.~Hammernik, T.~Klatzer, E.~Kobler, M.~P. Recht, D.~K. Sodickson, T.~Pock, and
  F.~Knoll, ``{Learning a variational network for reconstruction of accelerated
  MRI data},'' {\em Magnetic resonance in medicine}, vol.~79, no.~6,
  pp.~3055--3071, 2018.

\bibitem{uecker2010real}
M.~Uecker, S.~Zhang, D.~Voit, A.~Karaus, K.-D. Merboldt, and J.~Frahm,
  ``{Real-time MRI at a resolution of 20 ms},'' {\em NMR in Biomedicine},
  vol.~23, no.~8, pp.~986--994, 2010.

\bibitem{haji2018validation}
H.~Haji-Valizadeh, A.~A. Rahsepar, J.~D. Collins, E.~Bassett, T.~Isakova,
  T.~Block, G.~Adluru, E.~V. DiBella, D.~C. Lee, J.~C. Carr, {\em et~al.},
  ``{Validation of highly accelerated real-time cardiac cine MRI with radial
  k-space sampling and compressed sensing in patients at 1.5 T and 3T},'' {\em
  Magnetic resonance in medicine}, vol.~79, no.~5, pp.~2745--2751, 2018.

\bibitem{lim20193d}
Y.~Lim, Y.~Zhu, S.~G. Lingala, D.~Byrd, S.~Narayanan, and K.~S. Nayak, ``{3D
  dynamic MRI of the vocal tract during natural speech},'' {\em Magnetic
  resonance in medicine}, vol.~81, no.~3, pp.~1511--1520, 2019.

\bibitem{tamir2019unsupervised}
J.~I. Tamir, X.~Y. Stella, and M.~Lustig, ``Unsupervised deep basis pursuit:
  Learning reconstruction without ground-truth data,'' in {\em ISMRM annual
  meeting}, 2019.

\bibitem{huang2019deep}
P.~Huang, C.~Zhang, H.~Li, S.~K. Gaire, R.~Liu, X.~Zhang, X.~Li, and L.~Ying,
  ``{Deep MRI reconstruction without ground truth for training},'' in {\em
  ISMRM annual meeting}, 2019.

\bibitem{yaman2020self}
B.~Yaman, S.~A.~H. Hosseini, S.~Moeller, J.~Ellermann, K.~U{\u{g}}urbil, and
  M.~Ak{\c{c}}akaya, ``Self-supervised learning of physics-guided
  reconstruction neural networks without fully sampled reference data,'' {\em
  Magnetic resonance in medicine}, vol.~84, no.~6, pp.~3172--3191, 2020.

\bibitem{aggarwal2021ensure}
H.~K. Aggarwal, A.~Pramanik, and M.~Jacob, ``{ENSURE: Ensemble Stein’s
  unbiased risk estimator for unsupervised learning},'' in {\em IEEE
  International Conference on Acoustics, Speech and Signal Processing
  (ICASSP)}, pp.~1160--1164, 2021.

\bibitem{obungoloch2018design}
J.~Obungoloch, J.~R. Harper, S.~Consevage, I.~M. Savukov, T.~Neuberger,
  S.~Tadigadapa, and S.~J. Schiff, ``Design of a sustainable prepolarizing
  magnetic resonance imaging system for infant hydrocephalus,'' {\em Magnetic
  Resonance Materials in Physics, Biology and Medicine}, vol.~31, pp.~665--676,
  2018.

\bibitem{koonjoo2021boosting}
N.~Koonjoo, B.~Zhu, G.~C. Bagnall, D.~Bhutto, and M.~Rosen, ``{Boosting the
  signal-to-noise of low-field MRI with deep learning image reconstruction},''
  {\em Scientific reports}, vol.~11, no.~1, p.~8248, 2021.

\bibitem{schlemper2020deep}
J.~Schlemper, S.~S.~M. Salehi, C.~Lazarus, H.~Dyvorne, R.~O'Halloran,
  N.~de~Zwart, L.~Sacolick, J.~M. Stein, D.~Rueckert, M.~Sofka, {\em et~al.},
  ``{Deep learning MRI reconstruction in application to point-of-care MRI},''
  in {\em Proc. Intl. Soc. Mag. Reson. Med}, vol.~28, p.~0991, 2020.

\bibitem{xie2020noise2same}
Y.~Xie, Z.~Wang, and S.~Ji, ``Noise2same: Optimizing a self-supervised bound
  for image denoising,'' {\em Advances in Neural Information Processing
  Systems}, vol.~33, pp.~20320--20330, 2020.

\bibitem{lehtinen2018noise2noise}
J.~Lehtinen, J.~Munkberg, J.~Hasselgren, S.~Laine, T.~Karras, M.~Aittala, and
  T.~Aila, ``{Noise2Noise: Learning image restoration without clean data},''
  {\em arXiv preprint arXiv:1803.04189}, 2018.

\bibitem{batson2019noise2self}
J.~Batson and L.~Royer, ``Noise2self: Blind denoising by self-supervision,'' in
  {\em International Conference on Machine Learning}, pp.~524--533, PMLR, 2019.

\bibitem{millard2022theoretical}
C.~Millard and M.~Chiew, ``{A Theoretical Framework for Self-Supervised MR
  Image Reconstruction Using Sub-Sampling via Variable Density
  Noisier2Noise},'' {\em IEEE Transactions on Computational Imaging}, vol.~9,
  pp.~707--720, 2023.

\bibitem{moran2020noisier2noise}
N.~Moran, D.~Schmidt, Y.~Zhong, and P.~Coady, ``{Noisier2Noise: Learning to
  denoise from unpaired noisy data},'' in {\em Proceedings of the IEEE/CVF
  Conference on Computer Vision and Pattern Recognition}, pp.~12064--12072,
  2020.

\bibitem{desai2023noise2recon}
A.~D. Desai, B.~M. Ozturkler, C.~M. Sandino, R.~Boutin, M.~Willis,
  S.~Vasanawala, B.~A. Hargreaves, C.~R{\'e}, J.~M. Pauly, and A.~S. Chaudhari,
  ``{Noise2Recon: Enabling SNR-robust MRI reconstruction with semi-supervised
  and self-supervised learning},'' {\em Magnetic Resonance in Medicine}, 2023.

\bibitem{berk2008statistical}
R.~A. Berk {\em et~al.}, {\em Statistical learning from a regression
  perspective}, vol.~14.
\newblock Springer, 2008.

\bibitem{tian2020deep}
C.~Tian, L.~Fei, W.~Zheng, Y.~Xu, W.~Zuo, and C.-W. Lin, ``Deep learning on
  image denoising: An overview,'' {\em Neural Networks}, vol.~131,
  pp.~251--275, 2020.

\bibitem{krull2019noise2void}
A.~Krull, T.-O. Buchholz, and F.~Jug, ``{Noise2void-learning denoising from
  single noisy images},'' in {\em Proceedings of the IEEE/CVF Conference on
  Computer Vision and Pattern Recognition}, pp.~2129--2137, 2019.

\bibitem{huang2021neighbor2neighbor}
T.~Huang, S.~Li, X.~Jia, H.~Lu, and J.~Liu, ``Neighbor2neighbor:
  Self-supervised denoising from single noisy images,'' in {\em Proceedings of
  the IEEE/CVF conference on computer vision and pattern recognition},
  pp.~14781--14790, 2021.

\bibitem{hansen2015image}
M.~S. Hansen and P.~Kellman, ``Image reconstruction: an overview for
  clinicians,'' {\em Journal of Magnetic Resonance Imaging}, vol.~41, no.~3,
  pp.~573--585, 2015.

\bibitem{zeng2021review}
G.~Zeng, Y.~Guo, J.~Zhan, Z.~Wang, Z.~Lai, X.~Du, X.~Qu, and D.~Guo, ``{A
  review on deep learning MRI reconstruction without fully sampled k-space},''
  {\em BMC Medical Imaging}, vol.~21, no.~1, pp.~1--11, 2021.

\bibitem{wang2023k}
F.~Wang, H.~Qi, A.~De~Goyeneche, R.~Heckel, M.~Lustig, and E.~Shimron,
  ``{K-band: Self-supervised MRI Reconstruction via Stochastic Gradient Descent
  over K-space Subsets},'' {\em arXiv preprint arXiv:2308.02958}, 2023.

\bibitem{wiedemann2023deep}
S.~Wiedemann and R.~Heckel, ``A deep learning method for simultaneous denoising
  and missing wedge reconstruction in cryogenic electron tomography,'' {\em
  arXiv preprint arXiv:2311.05539}, 2023.

\bibitem{zbontar2018fastmri}
J.~Zbontar, F.~Knoll, A.~Sriram, T.~Murrell, Z.~Huang, M.~J. Muckley,
  A.~Defazio, R.~Stern, P.~Johnson, M.~Bruno, {\em et~al.}, ``{fastMRI: An open
  dataset and benchmarks for accelerated MRI},'' {\em arXiv preprint
  arXiv:1811.08839}, 2018.

\bibitem{lyu2023m4raw}
M.~Lyu, L.~Mei, S.~Huang, S.~Liu, Y.~Li, K.~Yang, Y.~Liu, Y.~Dong, L.~Dong, and
  E.~X. Wu, ``{M4Raw: A multi-contrast, multi-repetition, multi-channel MRI
  k-space dataset for low-field MRI research},'' {\em Scientific Data},
  vol.~10, no.~1, p.~264, 2023.

\bibitem{hendriksen2020noise2inverse}
A.~A. Hendriksen, D.~M. Pelt, and K.~J. Batenburg, ``{Noise2Inverse:
  Self-supervised deep convolutional denoising for tomography},'' {\em IEEE
  Transactions on Computational Imaging}, vol.~6, pp.~1320--1335, 2020.

\bibitem{Dabov2007}
K.~Dabov, A.~Foi, V.~Katkovnik, and K.~Egiazarian, ``{Image denoising by sparse
  3-D transform-domain collaborative filtering.},'' {\em IEEE transactions on
  image processing : a publication of the IEEE Signal Processing Society},
  vol.~16, pp.~2080--95, aug 2007.

\bibitem{Milla2020}
C.~Millard, A.~T. Hess, B.~Mailhe, and J.~Tanner, ``{Approximate Message
  Passing with a Colored Aliasing Model for Variable Density Fourier Sampled
  Images},'' {\em IEEE Open Journal of Signal Processing}, p.~1, 2020.

\bibitem{Virtue2017}
P.~Virtue and M.~Lustig, ``{The Empirical Effect of Gaussian Noise in
  Undersampled MRI Reconstruction},'' {\em Tomography}, vol.~3, pp.~211--221,
  dec 2017.

\bibitem{sriram2020end}
A.~Sriram, J.~Zbontar, T.~Murrell, A.~Defazio, C.~L. Zitnick, N.~Yakubova,
  F.~Knoll, and P.~Johnson, ``{End-to-end variational networks for accelerated
  MRI reconstruction},'' in {\em International Conference on Medical Image
  Computing and Computer-Assisted Intervention}, pp.~64--73, Springer, 2020.

\bibitem{ronneberger2015u}
O.~Ronneberger, P.~Fischer, and T.~Brox, ``U-net: Convolutional networks for
  biomedical image segmentation,'' in {\em International Conference on Medical
  Image Computing and Computer-Assisted Intervention}, pp.~234--241, Springer,
  2015.

\bibitem{kingma2014adam}
D.~P. Kingma and J.~Ba, ``{Adam: A method for stochastic optimization},'' {\em
  arXiv preprint arXiv:1412.6980}, 2014.

\bibitem{9433924}
B.~Yaman, S.~A.~H. Hosseini, S.~Moeller, J.~Ellermann, K.~Uğurbil, and
  M.~Akçakaya, ``Ground-truth free multi-mask self-supervised physics-guided
  deep learning in highly accelerated {MRI},'' in {\em 2021 IEEE 18th
  International Symposium on Biomedical Imaging (ISBI)}, pp.~1850--1854, 2021.

\bibitem{Wang2004}
Z.~Wang, A.~C. Bovik, H.~R. Sheikh, and E.~P. Simoncelli, ``{Image quality
  assessment: From error visibility to structural similarity},'' {\em IEEE
  Transactions on Image Processing}, vol.~13, no.~4, pp.~600--612, 2004.

\bibitem{zhao2022fastmri+}
R.~Zhao, B.~Yaman, Y.~Zhang, R.~Stewart, A.~Dixon, F.~Knoll, Z.~Huang, Y.~W.
  Lui, M.~S. Hansen, and M.~P. Lungren, ``{fastMRI+, Clinical pathology
  annotations for knee and brain fully sampled magnetic resonance imaging
  data},'' {\em Scientific Data}, vol.~9, no.~1, p.~152, 2022.

\bibitem{chen2022robust}
D.~Chen, J.~Tachella, and M.~E. Davies, ``Robust equivariant imaging: a fully
  unsupervised framework for learning to image from noisy and partial
  measurements,'' in {\em Proceedings of the IEEE/CVF Conference on Computer
  Vision and Pattern Recognition}, pp.~5647--5656, 2022.

\bibitem{chen2021equivariant}
D.~Chen, J.~Tachella, and M.~E. Davies, ``Equivariant imaging: Learning beyond
  the range space,'' in {\em Proceedings of the IEEE/CVF International
  Conference on Computer Vision}, pp.~4379--4388, 2021.

\bibitem{Stein1981}
C.~M. Stein, ``{Estimation of the Mean of a Multivariate Normal
  Distribution},'' {\em The Annals of Statistics}, vol.~9, pp.~1135--1151, nov
  1981.

\bibitem{10404226}
G.~P. Kumar, J.~Vijay~Arputharaj, P.~R. Kumar, D.~V. Kumar, B.~V.~V.
  Satyanarayana, and P.~R. Budumuru, ``A comprehensive review on image
  restoration methods due to salt and pepper noise,'' in {\em 2023 2nd
  International Conference on Automation, Computing and Renewable Systems
  (ICACRS)}, pp.~562--567, 2023.

\bibitem{spieker2023deep}
V.~Spieker, H.~Eichhorn, K.~Hammernik, D.~Rueckert, C.~Preibisch, D.~C.
  Karampinos, and J.~A. Schnabel, ``Deep learning for retrospective motion
  correction in mri: a comprehensive review,'' {\em IEEE Transactions on
  Medical Imaging}, 2023.

\end{thebibliography}

\end{document}